\documentclass[10pt,aps,prx,onecolumn,superscriptaddress,floatfix,nofootinbib]{revtex4-2}

\usepackage{amsmath}
\usepackage{amsthm}
\usepackage{amsfonts, amssymb, bbm, braket, accents}
\usepackage{graphicx}
\usepackage{wrapfig}
\usepackage[usenames,dvipsnames]{color}
\usepackage[makeroom]{cancel}
\usepackage{soul}
\usepackage{tikz}
\usepackage{float}
\usepackage{tcolorbox}
\makeatletter
\let\newfloat\newfloat@ltx
\makeatother
\usepackage{algorithm}
\usepackage{algpseudocode}
\usetikzlibrary{shapes.arrows} 
\usetikzlibrary{shapes.geometric}
\usetikzlibrary{calc} 
\usepackage[makeroom]{cancel}
\usepackage{verbatim}
\usepackage{pifont}
\usepackage[mathscr]{euscript}

\usepackage{booktabs}
\usepackage{xcolor}
\usepackage{siunitx}
\usepackage{colortbl}
\usepackage{algorithm}
\usepackage{algpseudocode}

\definecolor{cellmin}{rgb}{1,1,1}
\definecolor{cellmax}{rgb}{0.25,1,0.25}

\newtheorem{theorem}{Theorem}

\newtheorem{lemma}[theorem]{Lemma}
\newtheorem{corollary}[theorem]{Corollary}

\newtheorem{definition}[theorem]{Definition}

\newtheorem{remark}[theorem]{Remark}

\newcommand{\defeq}{\mathrel{:=}}

\newcommand{\norm}[1]{\left\lVert#1\right\rVert}
\newcommand{\abs}[1]{\left\lvert#1\right\rvert}
\newcommand{\eps}{\varepsilon}
\newcommand{\infi}{\infty}
\newcommand{\probP}{\text{I\kern-0.15em P}}
\newcommand{\probE}{\text{I\kern-0.15em E}}

\usepackage{graphicx}
\usepackage{enumitem}
\usepackage{xcolor}
\usepackage{tikz}
\usetikzlibrary{arrows.meta,positioning,fit}
\usepackage[breaklinks=true]{hyperref}

\definecolor{linkcolor}{HTML}{8B1E3F}
\definecolor{citecolor}{HTML}{1F4E79}
\definecolor{urlcolor}{HTML}{5E3C99}

\hypersetup{colorlinks=true,
    citecolor=citecolor,
    linkcolor=linkcolor,
    urlcolor=urlcolor
}

\usepackage[capitalise]{cleveref}
\crefformat{equation}{Eq.~(#2#1#3)}
\crefformat{section}{Sec.~#2#1#3}
\Crefformat{equation}{Equation~(#2#1#3)}
\crefformat{figure}{Fig.~#2#1#3}
\crefrangeformat{equation}{Eqs.~#3(#1)#4--#5(#2)#6}
\Crefformat{section}{Section~#2#1#3}

\begin{document}

\title{The Cost of Removing Tunability in Quantum Data Re-Uploading}

\author{Anthony Yuezhang Liu}
\affiliation{Centre for Quantum Technologies, National University of Singapore, Singapore}
\affiliation{Department of Mathematics, National University of Singapore, Singapore}

\author{Lirandë Pira}
\email{lpira@nus.edu.sg}
\affiliation{Centre for Quantum Technologies, National University of Singapore, Singapore}

\date{\today}

\begin{abstract}
Fixed encoding data re-uploading quantum circuits provide a striking example of universality emerging from a highly constrained architecture. However, universality alone is insufficient for assessing the theoretical and practical value of fixed and tunable upload circuits. The resource cost of removing tunability remains poorly understood. In this work, we establish quantitative depth-error scaling for approximating tunable upload circuits with fixed upload circuits. We show that a tunable upload circuit can be approximated by a fixed upload circuit using depth
\(
D = O_\sigma\!\left[(\log(1/\varepsilon))^\sigma\right]
\)
for every \(\sigma>1\), with a target dependent constant overhead, thereby improving the previously known polynomial dependence on \(1/\varepsilon\) with the same overhead.  Our proof is based on an auxiliary extension approximation mechanism that combines Gevrey class construction, Jackson's theorem and generalized quantum signal processing theorem.
Thus, the expressive power lost by removing tunability can be recovered using only polylogarithmic growth in circuit depth with a target dependent constant overhead. We further identify a periodic mismatch obstruction intrinsic to fixed upload approximations and use Turán-Nazarov inequalities to prove logarithmic lower bounds
\(
D = \Omega(\log(1/\varepsilon))
\)
for the approximation of mismatch class target tunable upload circuits. 
Conceptually, our analysis reveals two structural mechanisms underlying approximation in fixed upload architectures: auxiliary extensions and mismatch obstructions. These results provide a quantitative understanding of how expressivity is transferred from tunable frequencies into circuit depth, and suggest a broader framework for studying approximation complexity in quantum signal processing and related quantum learning models.
\end{abstract}

\maketitle

\newpage
\tableofcontents
\newpage

\section{Introduction}

A central theme in approximation theory and machine learning is that highly expressive models can emerge from remarkably restricted architectures. Classical universal approximation theorems show that simple neural network architectures can approximate arbitrary continuous functions despite severe structural constraints~\cite{cybenko1989approximation,hornik1991approximation}. Quantum computing has prompted the search for analogous expressive structures in quantum information processing, where computation is performed using quantum states and operations subject to quantum physical constraints~\cite{feynman1982simulating,shor1997,NielsenChuang2010}. Thus, similar phenomena have recently emerged in quantum machine learning, where parametrized quantum circuits and data re-uploading architectures exhibit universal approximation properties despite operating in low-dimensional Hilbert spaces and using highly constrained circuit structures~\cite{srinivasan2017surveyofquantumlearningtheory,benedetti2019parameterized,perezsalinas2020datareuploading,perezsalinas2021onequbit,goto2021universalapproximation,dutta2022singlequbit}. In the standard re-uploading paradigm, trainable quantum operations are interleaved with tunable upload gates. The expressive power of these models is closely connected to the ability to tune upload frequencies, allowing the circuit to explore a smaller but more problem-efficient approximation alphabet. The expressivity of different quantum machine learning model classes is a subject of extensive research~\cite{schuld2021effect,sim2019expressibility,du2020expressive,caro2020pseudodimension,caro2022generalization,manzano2025approximation}. Recently, however, it was shown that universality can be achieved under a significantly stronger architectural restriction: all tunable upload circuits may be approximated by circuits consisting of gates with one fixed upload frequency~\cite{perezsalinas2025universal}. In this fixed upload circuit model, tunable frequencies over the upload gates are removed entirely, yet, the resulting hypothesis class preserves the universal approximation property.

Universality alone, however, provides only a qualitative statement~\cite{barron1993universal,mhaskar1996neuralnetworks,yarotsky2017errorboundsapproximationsdeep}. Two architectures may represent the same function class while requiring vastly different resources to achieve a prescribed approximation accuracy. Consequently, the relevant question goes beyond whether fixed upload circuits are universal, and into the resource cost of universality in fixed upload circuits. More specifically, if tunable upload frequencies are removed from the model, it is of interest to understand how much circuit depth is required to recover the lost expressive power. Existing results give only a depth-error scale upper bound with a constant overhead and a polynomial dependence on \(\frac{1}{\varepsilon}\) for approximating single tunable upload gate~\cite{perezsalinas2025universal,yu2024nonasymptotic}. Little is known about quantitative depth-error scaling, the optimality and lower bounds of existing approximation schemes, or the approximation mechanisms responsible for their expressive power. From a practical perspective, fixed upload architectures may be attractive because they replace trainable encoding frequencies with a standardized encoding primitive. Such architectures can reduce the number of trainable parameters and may simplify implementation, calibration, or deployment in settings where data-upload operations are constrained. The key question is then whether this architectural simplification comes at a prohibitive expressive cost. Our results show that it does not: the expressive power lost by removing tunable encoding frequencies can be recovered with only polylogarithmic growth in circuit depth with target specific constant overhead.

In this work, we provide a sharper quantitative theory for approximating tunable upload circuits with fixed upload circuits. Our analysis reveals that the approximation process is governed by two complementary mechanisms developed around the rigidity of fixed upload circuits. The first is an \emph{auxiliary extension mechanism}, which enables efficient approximation by first restricting the target to a smaller domain on which uniform approximation is possible, and then re-extending artificially into the global domain with an auxiliary function compatible with the constraints of fixed upload architecture, enabling the application of classical approximation tools. The second is a \emph{mismatch obstruction mechanism}, which converts structural incompatibilities between the target and the approximating family into depth-error lower bounds. This works even if the global structural mismatch seems to be ignored by restricting target to smaller domain. These mechanisms lead to polylogarithmic depth-error upper bounds for all tunable upload circuits, and logarithmic depth-error lower bounds for mismatch class tunable upload circuits. Furthermore, these mechanisms can be used in more generalized settings, including the dense encoding fixed upload circuit structures in~\cite{perezsalinas2025universal}. Main contributions of this work can be summarized as follows.

First, we identify an auxiliary extension approximation mechanism for fixed upload circuits and combine Gevrey class extensions, Jackson's theorem and generalized quantum signal processing theorem to prove that a multivariate tunable upload circuit can be approximated by fixed upload circuits with depth
\begin{equation}  
    D= 
    O_\sigma\!\left(
    \sum_{j=1}^{L}\sum_{k=1}^{m}\abs{K_{jk}}
    +
    L\,m\left(\log\frac{L\,m\, C_\sigma}{\eps}\right)^\sigma
    \right), \sigma>1.
\end{equation} 
The growth is polylogarithmic with respect to \(1/\varepsilon\), with a target dependent constant overhead. This improves the previously known polynomial dependence on \(1/\varepsilon\) in~\cite{perezsalinas2025universal}'s Corollary 1 and extends to circuit wise theorem. Thus the expressive power lost by removing tunable upload frequencies can be recovered using only polylogarithmic growth in circuit depth, up to target specific constants. 

Second, we identify a mismatch obstruction intrinsic to fixed upload approximation targeting the mismatch class, and use Turán-Nazarov inequalities to prove logarithmic depth-error lower bounds
\begin{equation}
  D
    \ge
    \frac{1}{2\log\Lambda}
    \log\!\left(
    \frac{\Delta_4}{2\eps}
    \right)
    -
    \frac{S}{2},
\end{equation}
for mismatch class tunable upload circuits. To our knowledge, these are the first nontrivial lower bounds for fixed upload approximation.

Figure~\ref{fig:degree-sandwich} summarizes the resulting depth-error scaling. Intuitively, the upper bounds arise from constructing increasingly regular auxiliary extensions and realizing their approximants through generalized quantum signal processing (QSP)~\cite{low2017optimal,low2019hamiltonian,motlagh2024generalized,martyn2021grandunification,gilyen2019quantum}; while the lower bound follows from the impossibility of hiding a period mismatch using too few exponential modes. Under certain assumptions, these results provide a characterization of the resource cost of removing tunability in quantum data re-uploading models. From the perspective of quantum machine learning, our results clarify how expressivity is redistributed in fixed upload quantum architectures. Fixed upload circuits eliminate trainable upload frequencies while preserving universality, transferring the corresponding approximation burden into circuit depth. From the perspective of quantum signal processing and related approximations, our analysis identifies approximation mechanisms that may extend beyond fixed upload circuits and apply to a broader class of QSP-based constructions. Conceptually, our results show that removing tunable upload frequencies incurs only a modest cost: the lost expressive power can be recovered using fixed upload circuits with only polylogarithmic growth in depth, up to additive target dependent constants.

The remainder of the manuscript is organized as follows. Section~\ref{sec:background} fixes notation, defines the circuits, the generalized QSP theorem, and recalls approximation tools. Section~\ref{sec:mainresults} states the main depth-error results yielding upper and lower bounds as well as structural rigidity. Section~\ref{sec:tn-qsp} explains the two structural mechanisms behind these theorems, i.e. auxiliary extensions and mismatch obstruction. Section~\ref{sec:conclusion} concludes the manuscript. The appendices contain the rigidity proofs in Appendix~\ref{proof:rigidity}, upper-bound proofs and auxiliary-extension construction in Appendix~\ref{app:upper-bounds}, the lower-bound proofs in Appendix~\ref{app:lower-bounds}, and the multivariate function approximation corollaries in Appendix~\ref{app:function-approx}.

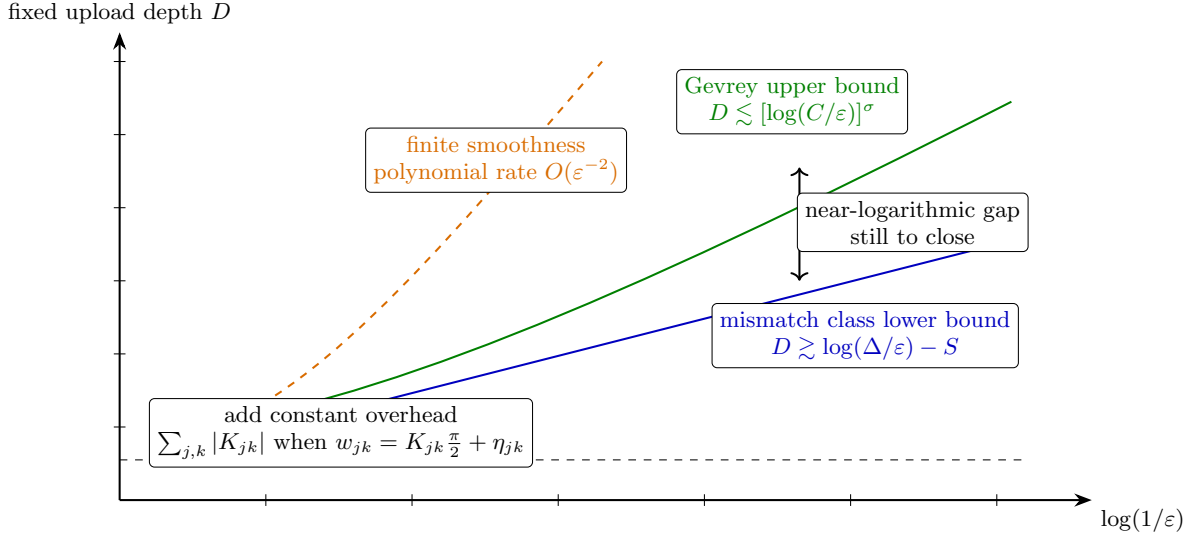
\begin{figure}[t]
\centering
\resizebox{0.88\linewidth}{!}{%
\begin{tikzpicture}[
    axis/.style={-{Stealth[length=2.4mm]}, thick},
    lower/.style={thick, blue!75!black},
    upper/.style={thick, green!50!black},
    finite/.style={thick, orange!85!black, dashed},
    band/.style={<->, thick},
    labelbox/.style={
      draw,
      rounded corners=2pt,
      fill=white,
      align=center,
      inner sep=3pt,
      font=\small
    }
]
\draw[axis] (0,0) -- (13.3,0) node[below right] {\(\log(1/\eps)\)};
\draw[axis] (0,0) -- (0,6.4) node[above] {fixed upload depth \(D\)};
\foreach \x in {2,4,6,8,10,12} {
  \draw (\x,0.08) -- (\x,-0.08);
}
\foreach \y in {1,2,3,4,5,6} {
  \draw (0.08,\y) -- (-0.08,\y);
}
\draw[lower] (0.8,0.65) -- (12.2,3.55);
\draw[upper] (0.8,0.95)
  .. controls (3.5,1.35) and (6.5,2.55) .. (12.2,5.45);
\draw[finite] (0.8,0.8)
  .. controls (2.6,1.3) and (4.2,3.2) .. (6.6,6.0);

\node[labelbox, text=blue!75!black] at (10.2,2.25)
  {mismatch class lower bound\\ \(D\gtrsim\log(\Delta/\eps)-S\)};
\node[labelbox, text=green!50!black] at (9.2,5.45)
  {Gevrey upper bound\\ \(D\lesssim[\log(C/\eps)]^\sigma\)};
\node[labelbox, text=orange!85!black] at (5.15,4.65)
  {finite smoothness\\polynomial rate \(O(\varepsilon^{-2})\)};

\draw[band] (9.3,3.0) -- (9.3,4.55);
\node[labelbox] at (10.85,3.80)
  {near-logarithmic gap\\ still to close};
\draw[dashed] (0,0.55) -- (12.4,0.55);
\node[labelbox, anchor=west] at (0.4,0.92)
  {add constant overhead\\ \(\sum_{j,k}|K_{jk}|\) when \(w_{jk}=K_{jk}\frac\pi2+\eta_{jk}\)};
\end{tikzpicture}%
}
\caption{depth-error scaling picture for the fixed upload approximation of tunable upload circuits. The mismatch obstruction argument gives a logarithmic lower bound when the target has a nonzero mismatch, while an auxiliary extension argument gives a polylogarithmic upper bound \(D=O([\log(1/\eps)]^\sigma)\) for every \(\sigma>1\). Finite smoothness auxiliary extension gives the weaker polynomial rate. Upper bound has a target specific constant overhead; Lower bound holds only for mismatch class target.}
\label{fig:degree-sandwich}
\end{figure}

\section{Preliminaries and Methods}\label{sec:background}
\subsection*{Notation}

We use the following notation throughout the manuscript. The scalar input is denoted by \(x\in[0,1]\) or \(t\in\mathbb R\); for multivariate inputs we write
\(
    X=[0,1]^m,
    x=(x_1,\ldots,x_m)\in X.
\)

A fixed upload gate is written as
\begin{equation}
    e^{i \frac{\pi}{2} x\sigma_z}.
\end{equation}
A tunable upload gate with frequency \(w\) is written as
\begin{equation}
    e^{iwx\sigma_z}.
\end{equation}
A residual tunable upload gate with residual frequency \(\eta\) is written as
\begin{equation}
    e^{i\eta x\sigma_z}.
\end{equation}

For a matrix \(M\), \(\norm{M}_{\mathrm{op}}\) denotes the operator norm and
\(\norm{M}_F\) the Frobenius norm. For a scalar function \(f\),
\(
    \norm{f}_\infi=\sup_x\abs{f(x)}
    \)
on the domain under discussion. For a compact domain \(Y\), \(C(Y)\) denotes the space of complex-valued continuous functions on \(Y\), equipped with the supremum norm. We write \(\mathcal C_{2\pi}^{(k)}\) for \(2\pi\)-periodic functions with \(k\) continuous derivatives, and \(\mathcal C_2^{(k)}\) for the corresponding 2-periodic class. The trigonometric polynomial spaces used in the approximation arguments are \(\mathcal L_n\) for \(2\pi\)-periodic trigonometric polynomials of order at most \(n\), and
\(
    \mathcal T_n^{(2)}
    =
    \left\{
    \sum_{\nu=-n}^{n}c_\nu e^{i\pi\nu x}:c_\nu\in\mathbb C
    \right\}
\)
for the 2-periodic normalization. 

For resource counts, \(L\) denotes the number of tunable layers in the target circuit \(U_L^w\), and \(m\) denotes the input dimension. Thus an \(m\)-variable circuit with \(L\) tunable layers contains \(L\,m\) tunable upload gates. The integer \(N\) denotes the fixed upload depth used to replace one residual tunable upload gate. The integer \(D\) denotes the total fixed upload depth, i.e., fixed upload gate count, of a fixed upload approximating circuit, and \(U_D^{mf}\) denotes an \(m\)-variable fixed upload circuit. In the one variable case, we write \(U_D^{1f}\). In the model definition below, \(M\) denotes a generic per-coordinate fixed upload repetition count. We use \(S\) for the number of exponential modes in the target circuit. Constants such as \(C_\sigma,c_\sigma,R_\sigma\) may depend on the Gevrey parameter \(\sigma>1\), but are independent of the target frequency \(w\) unless stated otherwise. We use \(O_\sigma(\cdot)\) to indicate that the hidden constant may depend on \(\sigma\).

\subsection{Circuit models and hypothesis families}

\begin{definition}[Hypothesis function]
Given a parametrized quantum circuit \(U_L(\theta,x)\) and a fixed computational reference state \(\ket{0}\), the associated hypothesis function is
\begin{equation}
    h_L(x)
    \defeq
    \langle 0|U_L(\theta,x)|0\rangle .
\end{equation}
\end{definition}
This is the convention used in Ref.~\cite[Eq.~(5)]{perezsalinas2025universal}.

\begin{definition}[Tunable upload circuits~\cite{perezsalinas2025universal,perezsalinas2021onequbit}]
Let \(x\in[0,1]^m\), let \(W=(w_1,\ldots,w_L)\) with \(w_j\in\mathbb R^m\), let \(\theta=(\theta_0,\ldots,\theta_L)\), let \(\phi=(\phi_0,\ldots,\phi_L)\), and let \(\lambda\in\mathbb R\). Following Ref.~\cite[Theorem~8]{perezsalinas2025universal}. We write the \(j\)-th tunable upload layer as
\(  A_j^w(x)
    \defeq
    e^{i\theta_j\sigma_z}
    e^{i\phi_j\sigma_y}
    e^{i\sigma_z w_j\cdot x},
\)
and a tunable upload circuit is written as:
\begin{equation}
    U_L^w(\theta,\phi,W,\lambda;x)
    \defeq
    \left(
    \prod_{j=1}^{L}A_j^w(x)
    \right)
    e^{i\theta_0\sigma_z}
    e^{i\phi_0\sigma_y}
    e^{i\lambda\sigma_z}.
\end{equation}

The corresponding output functions and tunable upload hypothesis family are
\begin{equation}
    h_L(x)
    \defeq
    \langle 0|
    U_L^w(\theta,\phi,W,\lambda;x)
    |0\rangle,
\end{equation}
\begin{equation}
    H_L^w
    \defeq
    \left\{
    h_L
    :
    \theta,\phi,W,\lambda\ \text{trainable}
    \right\}.
\end{equation}
\end{definition}

If \(w_j=(w_{j1},\ldots,w_{jm})\), then the data dependent part of the \(j\)-th layer factors as
\begin{equation}
    e^{i w_j\cdot x \sigma_z}
    =
    \prod_{k=1}^{m}e^{i w_{jk}x_{k}\sigma_z}.
\end{equation}
Thus an \(m\)-variable circuit with \(L\) tunable layers contains \(L\,m\) tunable upload gates.
This is the single qubit data re-uploading family with tunable upload gates in Ref.~\cite[Theorem~8]{perezsalinas2025universal} and\cite{perezsalinas2021onequbit}. For the rest of the manuscript, we refer to circuits of this model as tunable upload circuits. 

\begin{definition}[Fixed upload circuits~\cite{perezsalinas2025universal}]\label{fixeduploadcircuits}
Let \(x=(x_1,\dots,x_m) \in [0,1]^m\). Adapting from Ref.~\cite[Definition~2]{perezsalinas2025universal}, we write a single qubit fixed upload circuit as:
\begin{equation}
    U_D^{mf}(\theta,\phi,\lambda;x)
    \defeq
    \left(
    \prod_{j=1}^{L}\prod_{k=1}^{m}\prod_{q=1}^{M}
    e^{i\theta_{j,k,q}\sigma_z}
    e^{i\phi_{j,k,q}\sigma_y}
    e^{i \frac{\pi}{2} x_k\sigma_z}
    \right)
    e^{i\theta_0\sigma_z}
    e^{i\phi_0\sigma_y}
    e^{i\lambda\sigma_z}.
\end{equation}
The output functions and fixed upload hypothesis family are
\begin{equation}
    h_D(x)
    \defeq
    \langle 0|
    U_D^{mf}(\theta,\phi,\lambda;x)
    |0\rangle
\end{equation}
\begin{equation}
    H_D^f
    \defeq
    \left\{
    h_D
    :
    \theta,\phi,\lambda\ \text{trainable}
    \right\}.
\end{equation}
\end{definition}
Here \(D=L\,m\,M\) is the fixed upload gate count.
This is the single qubit fixed upload circuit family of Ref.~\cite[Definition~2]{perezsalinas2025universal}; its universality is stated in Ref.~\cite[Theorem~1]{perezsalinas2025universal}. We adapt the definition up to the structure and rescale the fix upload frequency for convenience.

Note that an inverse upload is counted with the same cost as a plain data gate, since
\(
    e^{-i \frac\pi 2 x\sigma_z}
    =
    X e^{i \frac{\pi}{2} x\sigma_z}X ,
\)
with the \(x\)-independent \(X\) rotations absorbed into adjacent trainable gates, as in Ref.~\cite[Theorem~6]{perezsalinas2025universal}. 

\begin{definition}[Dense fixed upload circuits~\cite{perezsalinas2025universal}]
\label{def:dense-fixed-family}
Let \(\{\ket{k}:1\le k\le m\}\) be the computational basis of an \(m\)-dimensional Hilbert space and define the dense fixed upload
\begin{equation}
    V_m(x)
    \defeq
    \sum_{k=1}^{m}
    \ket{k}\!\bra{k}\otimes e^{i\pi x_k\sigma_z/4}.
\end{equation}
For trainable unitaries \(R(\theta_j,\phi_j)\in SU(2m)\), expressed through Euler angle parameters, Ref.~\cite[Definition~3]{perezsalinas2025universal} defines dense fixed upload circuits as
\begin{equation}
    U_L(\Theta,\Phi,\lambda;x)
    \defeq
    \left(
    \prod_{j=1}^{L}R(\theta_j,\phi_j)V_m(x)
    \right)
    R(\theta_0,\phi_0)V_m(\lambda),
\end{equation}
The corresponding dense fixed upload hypothesis output function and family are
\begin{equation}
    h_L(x)
    \defeq
    \langle 0|
    U_L(\Theta,\Phi,\lambda;x)
    |0\rangle .
\end{equation}
\begin{equation}
    H_L
    \defeq
    \left\{
    h_L
    :
    \Theta,\Phi,\lambda\ \text{trainable}
    \right\}.
\end{equation}
\end{definition}
This is the multiqubit dense encoding fixed upload family of Ref.~\cite[Definition~3]{perezsalinas2025universal}; its universality is stated in Ref.~\cite[Theorem~2]{perezsalinas2025universal}.

\subsection{Generalized QSP theorem}

\begin{theorem}[Generalized QSP for fixed upload circuits \cite{motlagh2024generalized,perezsalinas2025universal}]
\label{thm:gqsp-background}
\label{thm:gqsp-realization}
Let \(P(z)=\sum_{n=-N}^{N}p_nz^n\) be a Laurent polynomial satisfying
\begin{equation}
    \sup_{\abs{z}=1}\abs{P(z)}\le 1 .
\end{equation}
Then there exists a one variable fixed upload circuit using \(O(N)\) copies of the fixed upload gate \(e^{i \frac{\pi}{2} x\sigma_z}\) whose unitary has the form
\begin{equation}
    R_N(x)=
    \begin{pmatrix}
    P(e^{i\pi x}) & -Q(e^{i\pi x})^*\\
    Q(e^{i\pi x}) & P(e^{i\pi x})^*
    \end{pmatrix},
\end{equation}
where \(Q\) is a Laurent polynomial and
\begin{equation}
    \abs{P(e^{i\pi x})}^2+\abs{Q(e^{i\pi x})}^2=1
    \qquad \text{for all }x .
\end{equation}
\end{theorem}

Motlagh and Wiebe's generalized QSP theorem gives necessary and sufficient conditions for the quantum circuit realization of polynomial pairs \(P,Q\) of degree \(d\) satisfying \(\abs{P}^2+\abs{Q}^2=1\) on the unit circle \cite[Theorem~3]{motlagh2024generalized}. In the same reference, Corollary~5 states that any polynomial \(P\) bounded by one on the unit circle admits such a companion \(Q\), and Theorem~6 gives the shifted Laurent form with negative powers. Ref.~\cite[Theorem~6]{perezsalinas2025universal} rewrites this generalized-QSP product as a fixed upload circuit, and the proof of Ref.~\cite[Theorem~5]{perezsalinas2025universal} uses the same realization to approximate a tunable upload gate. Applying the shifted Laurent form to \(z^NP(z)\) yields the stated \(O(N)\)-upload fixed upload circuit. The explicit gate count is \(2N\) fixed upload gates and \(2N+1\) trainable unitaries, but we will only track the asymptotic scaling in \(N\).

In generalized QSP theorem the parity constraint is removed with operations based on a one sided signal gate \(A = diag(U,I)\). In principle, if we were to conclude that any arbitrary QSP type unitary with upper left entry \(P(e^{i\pi x})\) can be realized by fixed upload circuit, we need to first conclude that a one sided signal \(diag(e^{i\pi x},1)\) is realizable by this circuit. For this purpose, \cite[Theorem~6]{perezsalinas2025universal} gave a method to produce \(diag(e^{ix},1)\) with a global phase gate \(e^{ix/2}\) or \(e^{-ix/2}\), compensated by querying \(L\) times the operator absorbed in the gates before and after\cite{perezsalinas2025universal}, along with a single data gate \(e^{ix/2 \sigma_z}\). We rescale this argument to one based on \(e^{i \frac{\pi}{2}x \sigma_z}\), producing the desired one sided signal \(diag(e^{i\pi x},1)\); therefore the generalized QSP theorem holds for \(R_N(x)\). Therefore, any QSP block with Laurent polynomials \(P(e^{i \pi x})\) on the upper left entry is obtainable with fixed upload circuits.

It should be noted that all fixed upload circuits are indeed generalized QSP circuits with matrix entries \(P(e^{i \frac\pi2 x})\) by definition, while their \(P\) has a parity constraint, because of the limitation of the signal gate, in our terms the fixed upload gate. While all generalized QSP circuits with matrix entries \(P(e^{i \pi x})\) are realizable by fixed upload circuits as defined here, it is only attainable with a certain parity condition over circuit depth. This will be discussed in \ref{rem:parity}. This should be treated carefully because in upper bound theorems we use the latter part as a realization theorem that admits desired polynomials with fixed upload circuits; while in single gate lower bound proof we use the former part that describes the exponential polynomial nature of fixed upload circuits. The constant trainable matrices can be chosen in the exact \(SU(2)\) Euler form as defined \cite{perezsalinas2025universal,motlagh2024generalized}.

\subsection{Approximation tools}

\begin{theorem}[Jackson IV theorem {\cite[Theorem~16.4]{powell1981approximation}}]
\label{thm:jackson}
Let \(f\in\mathcal C_{2\pi}^{(k)}\). Let \(\mathcal L_n\) denote the trigonometric polynomials of order at most \(n\), and define the best approximation error
\begin{equation}
    E_n(f)
    \defeq
    \inf_{T\in\mathcal L_n}
    \norm{f-T}_\infi .
\end{equation}
Then
\begin{equation}
    E_n(f)
    \le
    \left(\frac{\pi}{2n+2}\right)^k
    \norm{f^{(k)}}_\infi .
\end{equation}
\end{theorem}

For the \(f\in\mathcal C_2^{(k)}\) normalization used below, apply \cref{thm:jackson} to
\(g(y)=f(y/\pi)\), where \(g\in\mathcal C_{2\pi}^{(k)}\). If \(f\in\mathcal C_2^{(k)}\), then we have
\(
    \mathcal T_n^{(2)}
    \defeq
    \left\{
    \sum_{\nu=-n}^{n}c_\nu e^{i\pi\nu x}:c_\nu\in\mathbb C
    \right\}.
\)
Therefore
\begin{equation}
    \inf_{T\in\mathcal T_n^{(2)}}
    \norm{f-T}_\infi
    \le
    \frac{1}{(2n+2)^k}
    \norm{f^{(k)}}_\infi.
\end{equation}

\begin{definition}[2-periodic Gevrey-\(\sigma\) class \cite{KuhnPetersen2022PeriodicGevrey}]
\label{def:2-periodic-gevrey}
Let \(\sigma>1\). A 2-periodic smooth function \(A\in C_2^\infty[0,2]\)
belongs to \(G^\sigma_2\) if there exist constants \(C,R>0\) such that
\begin{equation}
    \norm{A^{(k)}}_\infi
    \le
    C R^k(k!)^\sigma,
    \qquad \forall k \in \mathbb{N} \cup \{0\}
\end{equation}
Notice that \(G^\sigma_2\subset C_2^r[0,2]\) for every \(r\in\mathbb N\).
\end{definition}

Gevrey classes were introduced in 1918~\cite{Gevrey1918}. Since then, these function classes are considered as one of the most important intermediate classes between smooth and analytic functions and have played an important role in differential equations. For a systematic treatment of Gevrey classes, we refer to Rodino's~\cite{Rodino1993Gevrey}. In this work we only use Gevrey class reduced to single variable case. Our definition mostly adopts the convention used in~\cite{KuhnPetersen2022PeriodicGevrey}, but on 2-torus, reduced to single variable case. For the stability of Gevrey class under certain operations, especially inverse closedness, we refer to~\cite{RainerSchindl2014Composition}.

\begin{theorem}[Turán-Nazarov inequality {\cite{nazarov1993local}}]
\label{thm:turan-nazarov}
Let
\begin{equation}
    p(x)=\sum_{j=1}^{s}a_je^{i\lambda_jx},
    \qquad \lambda_j\in\mathbb R,
\end{equation}
be an exponential polynomial with at most \(s\) terms. There is a universal
constant \(\Lambda_0>1\) such that, for every interval \(I\subset\mathbb R\)
and every measurable \(E\subset I\) of positive measure,
\begin{equation}
    \sup_{x\in I}\abs{p(x)}
    \le
    \left(\frac{\Lambda_0\abs I}{\abs E}\right)^{s-1}
    \sup_{x\in E}\abs{p(x)} .
\end{equation}
\end{theorem}

For the proof of this theorem, we refer to Theorem I in Nazarov's original work~\cite{nazarov1993local}. In this work, we are only concerned with exponential polynomials with pure imaginary frequencies; therefore the original inequality simplifies to the form stated above, as a result of \(Re(i\lambda_j)=0\) for all \(j\).

\section{Main Results}\label{sec:mainresults}

\subsection{Structural rigidity of fixed upload circuits}
\begin{theorem}[Rigidity of fixed upload circuits]
\label{thm:rigidity}
Let \(U_D^{mf}\) be a fixed upload circuit as defined in \ref{fixeduploadcircuits}, then
\begin{equation}
    U_D^{mf}(0)
    =
    U_D^{mf}(4)
    =
    (-1)^DU_D^{mf}(2)
\end{equation}
Hence universal approximation of tunable upload circuits family on intervals containing \([0,2]^m\) is not possible for fixed upload circuits family.
\end{theorem}

The proof is given in Appendix subsection \ref{proof:rigidity}. This property underlies the structural analysis developed in this manuscript. The auxiliary extension mechanism arises from techniques used to circumvent this constraint. Mismatch obstruction is an application of this property. 

The above mentioned rigidity comes from the periodicity and parity structure native to the fixed upload circuit. 

\begin{remark}[Parity of fixed upload Fourier modes]
\label{rem:parity}
Let \(z=e^{i\pi x/2}\). A one variable depth-\(D\) fixed upload circuit has the form
\begin{equation}
    U_D^{f}(x)
    =
    B_De^{i\pi x\sigma_z/2}B_{D-1}e^{i\pi x\sigma_z/2}\cdots
    B_1e^{i\pi x\sigma_z/2}B_0 ,
\end{equation}
where the matrices \(B_j\) are independent of \(x\). Since
\begin{equation}
    e^{i\pi x\sigma_z/2}
    =
    z\ket{0}\!\bra{0}
    +
    z^{-1}\ket{1}\!\bra{1},
\end{equation}
expanding the product gives
\begin{equation}
    U_D^f(x)
    =
    \sum_{\alpha\in\{0,1\}^D}
    z^{D-2N_1(\alpha)}
    B_D P_{\alpha_D} B_{D-1}P_{\alpha_{D-1}}
    \cdots
    B_1P_{\alpha_1}B_0 ,
\end{equation}
where
\(
    P_0=\ket{0}\!\bra{0},
    P_1=\ket{1}\!\bra{1},
\)
and \(N(\alpha)\) denotes the number of indices \(j\) for which
\(\alpha_j=1\). Indeed, for \(\alpha =1\) the upload contributes \(-1\) power, canceling a positive power out. Thus, if \(N(\alpha)=r\), the corresponding power of \(z\) is \(D-2r\). Grouping together all terms with the same value of \(r\), we get
\begin{equation}
    U_D^f(x)
    =
    \sum_{r=0}^{D} z^{D-2r} C_r ,
\end{equation}
for some \(x\)-independent matrices \(C_r\). Therefore every scalar matrix element of \(U_D^f(x)\) is a Laurent polynomial of the form
\begin{equation}
    \langle a|U_D^f(x)|b\rangle
    =
    \sum_{r=0}^{D}
    c_{ab,r} z^{D-2r}
    =
    \sum_{\substack{n=-D\\ n\equiv D\;(\mathrm{mod}\;2)}}^{D}
    c_{ab,n} z^n
    =
    \sum_{\substack{n=-D\\ n\equiv D\;(\mathrm{mod}\;2)}}^{D}
    c_{ab,n} e^{i\pi n x/2}.
\end{equation}
Hence a depth \(D\) fixed upload circuit has at most \(D+1\) possible Fourier modes, and all exponents have the same parity as \(D\)
\begin{equation}
    n\equiv D \pmod 2 .
\end{equation}
In particular, even depth circuits contain only even powers of \(z\), hence Laurent polynomial in \(z^2=e^{i\pi x}\). Odd depth circuits contain only odd powers of \(z\), hence \(z\) times Laurent polynomial in \(z^2=e^{i\pi x}\).
\end{remark}

Consequently, even depth fixed upload circuits have scalar entries in Laurent polynomials of \(e^{i\pi x}\), while odd depth fixed upload circuits have scalar entries in \(e^{i\pi x/2}\) times Laurent polynomials of \(e^{i\pi x}\). Thus a generalized QSP realization written purely in the variable \(e^{i\pi x}\) corresponds to the even parity fixed upload circuits.

This implies that in order to realize arbitrary Laurent polynomial over \(e^{i\pi x}\) by a fixed upload circuit, we need upload gates with frequencies at least finer than \(\frac\pi2\). The theoretical realizable set is slightly larger than mere Laurent polynomial entried matrix over \(e^{i\pi x}\), as a odd parity circuit leaves a fixed upload gate to the end.

\subsection{Upper bounds from auxiliary extensions}

Restricting our target interval to \([0,1]^m\), we can prove the universality and gain quantitative depth-error upper bounds for various cases.

\begin{theorem}[Single tunable upload gate upper bound]
\label{thm:main-gevrey-upper}
Let \(w\in\mathbb R\), and choose
\begin{equation}
    w=K\frac\pi2+\eta,
    \qquad K\in\mathbb Z,\qquad
    \abs{\eta}\le\frac{\pi}{4}.
\end{equation}
Fix \(\sigma>1\). There exist constants \(C_\sigma,c_\sigma>0\), independent of \(w\), such that for every \(N\ge1\) there is a fixed upload circuit \(U_D^{1f}\), using \(D=O(\abs K+N)\) fixed upload gates, with
\begin{equation}
    \sup_{x\in[0,1]}
    \norm{U_D^{1f}(x)-e^{iw x\sigma_z}}_F
    \le
    C_\sigma\exp(-c_\sigma N^{1/\sigma}).
\end{equation}
Consequently, approximation error at most \(\eps\in(0,1)\) is achieved with
\begin{equation}
    D =
    O_\sigma\!\left(
    \abs{K}+
    \left[\log\!\left(\frac{C_\sigma}{\eps}\right)\right]^\sigma
    \right)
\end{equation}
fixed upload gates, with \(\abs{K}\) being a constant overhead.
\end{theorem}

\noindent The proof is given in Appendix subsection~\ref{proof:main-gevrey-upper}. In this theorem, we study how a single tunable upload gate is approximated by a fixed upload circuit. 

\begin{theorem}[circuit wise upper bound]
\label{thm:main-circuit-upper}
Let \(U_L^w\) be an \(m\)-variable tunable upload circuit with \(L\) tunable layers. Write \(w_j=(w_{j1},\ldots,w_{jm})\), and choose
\begin{equation}
    w_{jk}=K_{jk}\frac\pi2+\eta_{jk},
    \qquad K_{jk}\in\mathbb Z,\qquad
    \abs{\eta_{jk}}\le\frac{\pi}{4}.
\end{equation}
For every \(\sigma>1\), replacing each tunable upload gate \(e^{iw_{jk}x_k\sigma_z}\) by the subcircuit given in \(\Cref{thm:main-gevrey-upper}\) gives a fixed upload circuit \(U_D^{mf}\) satisfying
\begin{equation}
    \sup_x
    \norm{U_L^w(x)-U_D^{mf}(x)}_{\mathrm{op}}
    \le
    L\,m\, C_\sigma e^{-c_\sigma N^{1/\sigma}},
\end{equation}
using
\begin{equation}
    D=
    O\!\left(
    \sum_{j=1}^{L}\sum_{k=1}^{m}\abs{K_{jk}}+L\,m\,N
    \right)
\end{equation}
fixed upload gates. Hence approximation error at most \(\eps\in(0,1)\) is obtained with
\begin{equation}
    D= 
    O_\sigma\!\left(
    \sum_{j=1}^{L}\sum_{k=1}^{m}\abs{K_{jk}}
    +
    L\,m\left(\log\frac{L\,m\, C_\sigma}{\eps}\right)^\sigma
    \right).
\end{equation}
Here \(\sum_{j=1}^{L}\sum_{k=1}^{m}\abs{K_{jk}}\) is a constant overhead dependent only on the specific tunable upload circuit.
\end{theorem}

\noindent The proof is given in Appendix subsection~\ref{proof:main-circuit-upper}. 

Therefore, for any given target tunable upload circuit, the depth of fixed upload circuit it takes to approximate, with respect to error, grows like
\begin{equation}
    D=O_{\sigma}\!\left[\left(\log\frac{1}{\varepsilon}\right)^{\sigma}\right].
\end{equation}

Consequently, the resource it takes to remove tunability grows at a polylogarithmic rate with respect to \(\frac{1}{\varepsilon}\).

\begin{corollary}[Dense fixed upload circuit upper bound]
\label{cor:dense-circuit-upper}
\label{cor:measured-circuit-upper}
Let \(X=[0,1]^m\), and let
\(
    U_D^{\mathrm{dense}}(\Theta,\Phi,\lambda;x)
\)
denote a circuit using \(D\) dense fixed upload calls from the dense fixed upload family in
\(\ref{def:dense-fixed-family}\), equivalently the architecture of
Ref.~\cite[Theorem~2]{perezsalinas2025universal}. Write
\(
    h_U(x)=\langle0|U(x)|0\rangle .
\)
Suppose \(U_L^w\) is an \(m\)-variable tunable upload circuit with \(L\) tunable layers and coordinate frequencies \(w_{jk}=K_{jk}\frac\pi2+\eta_{jk}\), \(\abs{\eta_{jk}}\le\pi/4\) with output function \(h\). For every \(\sigma>1\) and every \(\eps\in(0,1)\), there exists a dense fixed upload circuit \(U_D^{\mathrm{dense}}\) such that
\begin{equation}
    \norm{h-h_{U_D^{\mathrm{dense}}}}_\infi\le\eps
\end{equation}
and total dense upload count
\begin{equation}
    O_\sigma\!\left(
    \sum_{j=1}^{L}\sum_{k=1}^{m}\abs{K_{jk}}
    +
    L\,m\left(\log\frac{L\,m\, C_\sigma}{\eps}\right)^\sigma
    \right).
\end{equation}
\end{corollary}

\noindent The proof is given in Appendix subsection~\ref{proof:dense-circuit-upper}. This is an immediate extension of \cref{thm:main-circuit-upper}, we refer to~\cite{perezsalinas2025universal} for a detailed discussion on the direct connection between dense fixed upload circuits and fixed upload circuits.

\begin{corollary}[Bridging towards continuous class]
\label{cor:continuous-function-upper}
\label{thm:dense-fixed-upper}
\label{cor:all-continuous-function-upper}
Let \(X=[0,1]^m\), and define
\(
    C_{\le1}(X)=\{f\in C(X):\norm f_\infi\le1\}.
\)
Let \(\mathcal F\subset C_{\le1}(X)\), and assume that tunable upload hypothesis functions \(h\) are dense in \(\mathcal F\). Then both fixed upload circuits and dense fixed upload circuits hypothesis functions are dense in \(\mathcal F\).

More quantitatively, suppose \(f\in\mathcal F\) and an \(m\)-variable tunable upload circuit \(U_L^w\) with \(L\) tunable layers and coordinate frequencies \(w_{jk}=K_{jk}\frac\pi2+\eta_{jk}\), \(\abs{\eta_{jk}}\le\pi/4\), satisfies
\begin{equation}
    \norm{f-h}_\infi\le\frac{\eps}{2}.
\end{equation}
For every \(\sigma>1\), there exist a coordinate-wise fixed upload circuit \(U_D^{mf}\) and a dense fixed upload circuit \(U_D^{\mathrm{dense}}\) such that
\begin{equation}
    \norm{f-h_{U_D^{mf}}}_\infi\le\eps,
    \qquad
    \norm{f-h_{U_D^{\mathrm{dense}}}}_\infi\le\eps ,
\end{equation}
with fixed upload and dense-upload counts both bounded by
\begin{equation}
    O_\sigma\!\left(
    \sum_{j=1}^{L}\sum_{k=1}^{m}\abs{K_{jk}}
    +
    L\,m\left(\log\frac{2L\,m\, C_\sigma}{\eps}\right)^\sigma
    \right).
\end{equation}
In particular, taking \(\mathcal F=C_{\le1}(X)\) gives the results towards continuous functions.
\end{corollary}

\noindent The proof is given in Appendix subsection~\ref{proof:continuous-function-upper}.

This corollary exhibits that the removal of tunability preserves the universality in continuous function class. It also gives a partial upper bound of the continuous class approximation by fixed upload circuits. But the limitation of this corollary is that it only accounts for how approximation works for the intermediate tunable circuit proxy. Further study can investigate whether it is possible to eliminate this proxy dependency, while maintaining the polylogarithmic depth-error scaling.

\subsection{Lower bounds from mismatch obstruction}

\begin{theorem}[Single residual tunable upload gate lower bound]
\label{thm:main-single-gate-lower}
  Let \(\eta\notin\frac\pi2\mathbb Z\). Assume a residual tunable upload gate is approximated by a fixed upload circuit with
\begin{equation}
    \sup_{x\in[0,1]}
    \norm{U_N^{1f}(x)-e^{i\eta x\sigma_z}}_F
    \le \eps .
\end{equation}
If \(\eps\le \sqrt2\,\sin^2 2\eta\), then
\begin{equation}
    \eps
    \ge
    2\abs{\sin2\eta}\,\Lambda^{-(2N+1)}
\end{equation}
for a universal constant \(\Lambda>1\). Equivalently,
\begin{equation}
    N
    \ge
    \frac{1}{2\log\Lambda}
    \log\!\left(\frac{2\abs{\sin2\eta}}{\eps}\right)
    -\frac12 .
\end{equation}
In particular, for fixed \(\eta\notin \frac\pi2\mathbb Z\),
\begin{equation}
    N=\Omega_\eta\!\left(\log\frac1\eps\right).
\end{equation}
\end{theorem}

\noindent The proof is given in Appendix subsection~\ref{proof:main-single-gate-lower}. 

\begin{theorem}[No free circuit wise lower bound]

\label{thm:main-no-mismatch-free-bound}
There is no positive lower bound depending only on the number of tunable upload gates, their frequencies, and the target error \(\eps\) that holds for every target tunable upload circuit.
\end{theorem}

\noindent The proof is given in Appendix subsection~\ref{proof:main-no-mismatch-free-bound}. 

\begin{definition}[4-periodic mismatch class]
For an arbitrary one variable tunable upload circuit \(U_L^w:\mathbb R\to U(d)\), define
\begin{equation}
    \Delta_4(U_L^w)
    \defeq
    \norm{U_L^w(4)-U_L^w(0)}_{\mathrm{op}} .
\end{equation}
\(\Delta_4\) measures how badly the target circuit violates the continuous 4-periodic condition forced by the periodicity induced from the \(\frac{\pi}{2}\)-frequencies of fixed upload. 
\end{definition}
This criterion describes a large class of operators that are not 4-periodic. Intuitively, this can be seen as a generalization of the residualness of \(e^{i\eta x \sigma_z}\). We define it only on the one variable case because a multivariate generalization is immediate by coordinate wise one variable witnesses. We will discuss this generalization later in this manuscript.

\begin{theorem}[Mismatch class circuit wise lower bound]
\label{thm:main-circuit wise-lower}
  Let \(U_L^w:\mathbb R\to U(2)\) be a target tunable upload circuit. Assume that every scalar matrix element of \(U_L^w(t)\) is an exponential polynomial with at most \(S\) terms. In particular, an one variable tunable upload circuit with \(L\) tunable layers, equivalently \(L\) tunable upload gates, satisfies this assumption with \(S\le 2^L\) by \(\Cref{lem:finite-phase-term-count}\).

Let \(U_D^{1f}(t)\) be any one variable fixed upload circuit with fixed upload depth at most \(D\). Suppose
\begin{equation}
    \sup_{t\in[0,1]}
    \norm{U_L^w(t)-U_D^{1f}(t)}_{\mathrm{op}}
    \le \eps .
\end{equation}
If
\begin{equation}
    0<\eps\le\frac{\Delta_4(U_L^w)}{2},
\end{equation}
then
\begin{equation}
    D
    \ge
    \frac{1}{2\log\Lambda}
    \log\!\left(
    \frac{\Delta_4(U_L^w)}{2\eps}
    \right)
    -
    \frac{S}{2},
\end{equation}
where \(\Lambda>1\) is a universal constant.

In particular, if \(U_L^w\) is fixed, \(S\) is fixed, and \(\Delta_4(U_L^w)>0\), then every fixed upload approximation must have
\begin{equation}
    D=\Omega\!\left(\log\frac1\eps\right).
\end{equation}
\end{theorem}

\noindent The proof is given in Appendix subsection~\ref{proof:main-circuit wise-lower}. Notice that the lower bound contains a constant term overhead dependent on the target, that scales with \(L\), the layer of the target tunable circuit. In the most general case it scales at most exponentially.

\begin{corollary}[Multivariate lower bounds by coordinate slice]
\label{coordinateslice}
Even though the theorem is stated for one variable, the same lower bound applies to multivariate circuits with a slightly generalized mismatch condition.

If \(U_L^w:\mathbb R^m\to U(d)\) and \(U_D^{mf}:\mathbb R^m\to U(d)\), take any \(\gamma(t)=x^{(0)}+te_k\). Uniform approximation on \([0,1]^m\) implies approximation of \(U_L^w\circ\gamma\) by the restricted one variable fixed upload circuit \(U_D^{1f}\defeq U_D^{mf}\circ\gamma\) on \([0,1]\). Fixed upload gates in variables other than \(x_k\) become \(t\)-independent unitaries on this slice, while fixed upload gates in \(x_k\) give the same 4-periodic Fourier structure as above. 

Therefore \(\Cref{thm:main-circuit wise-lower}\) applies to every multivariate tunable circuit with coordinate slice mismatch with \(\Delta_4(U_L^w\circ\gamma)>0\). Therefore, for multivariate approximation, we need a mismatch condition over these coordinate slice. 

Intuitively, the \(\gamma\) composition mismatch is a generalization of from one dimensional case to multidimensional case. Continuous periodicity over \([0,4]^m\) requires the boundaries to agree. A coordinate slice mismatch witnesses the mismatch of the boundary. Therefore, nonzero multidimensional mismatch is equivalent to the existence of a bad coordinate slice, represented by \(\gamma\).
\end{corollary}

\begin{corollary}[Mismatch class exponential polynomial function lower bound]
\label{cor:mismatch-function-lower}
Let \(h:\mathbb R\to\mathbb C\) be an exponential polynomial with at most \(S\) terms and 4-periodic mismatch
\begin{equation}
    \Delta_4(h)
    \defeq
    \abs{h(4)-h(0)}
    >0 .
\end{equation}
Let \(h_D\) be any hypothesis function produced by a fixed upload circuit of depth at most \(D\). Equivalently for the present bound,
\begin{equation}
    h_{D}(t)=     
    \langle 0|
    U_D^{mf}(t)
    |0\rangle
    =\sum_{n=-D}^{D}a_ne^{i\frac\pi2 nt}.
\end{equation}
If
\begin{equation}
    \sup_{t\in[0,1]}\abs{h(t)-h_D(t)}
    \le\eps,
    \qquad
    0<\eps\le\frac{\Delta_4(h)}{2},
\end{equation}
then
\begin{equation}
    D
    \ge
    \frac{1}{2\log\Lambda}
    \log\!\left(
    \frac{\Delta_4(h)}{2\eps}
    \right)
    -
    \frac{S}{2},
\end{equation}
where \(\Lambda>1\) is a universal constant.
\end{corollary}

\noindent The proof is given in Appendix subsection~\ref{proof:mismatch-function-lower}.

\begin{corollary}[Dense fixed upload lower bound]
\label{cor:dense-family-lower}
Let \(h^*:\mathbb R^m\to\mathbb C\) be a target function and fix a coordinate slice
\begin{equation}
    \gamma(t)=x^{(0)}+te_k,
    \qquad
    x^{(0)}\in[0,1]^m,
    \qquad
    x^{(0)}_k=0,
\end{equation}
where \(e_k\) is the \(k\)-th standard basis vector. Assume that
\(
    h^*\circ\gamma
\)
is an exponential polynomial with at most \(S\) terms and has nonzero 8-periodic mismatch
\begin{equation}
    \Delta_{8,\gamma}(h^*)
    \defeq
    \abs{h^*(\gamma(8))-h^*(\gamma(0))}
    >0 .
\end{equation}
Let \(h_{U_D^{\mathrm{dense}}}\) be any scalar amplitude hypothesis produced by a dense fixed upload circuit of \(\Cref{def:dense-fixed-family}\) using at most \(D\) dense uploads. If
\begin{equation}
    \sup_{x\in[0,1]^m}
    \abs{h^*(x)-h_{U_D^{\mathrm{dense}}}(x)}
    \le\eps,
    \qquad
    0<\eps\le\frac{\Delta_{8,\gamma}(h^*)}{2},
\end{equation}
then
\begin{equation}
    D
    \ge
    \frac{1}{2\log\Lambda}
    \log\!\left(
    \frac{\Delta_{8,\gamma}(h^*)}{2\eps}
    \right)
    -
    \frac{S}{2},
\end{equation}
where \(\Lambda>1\) is a universal constant. Consequently, no dense upload circuits of depth \(D(\eps)=o(\log(1/\eps))\) can uniformly approximate every member of any function class containing such an \(h^*\).
\end{corollary}

\noindent The proof is given in Appendix subsection~\ref{proof:dense-family-lower}.

This corollary combines Corollary~\ref{coordinateslice}, Corollary~\ref{cor:mismatch-function-lower} and the mechanism of \cref{thm:main-circuit wise-lower}. We could see this as an example of application of the mismatch obstruction mechanism to more general problems. 

\section{Fixed Upload Approximation Mechanisms}\label{sec:tn-qsp}

This section discusses the approximation mechanisms we observed behind the upper and lower bounds. Moreover, we see a possibility to extend these structural principles into a general framework that can be useful beyond the examples presented in this work. The fixed upload imposes a rigid global structure: after setting \(z=e^{i\frac\pi2 x},\) single variable fixed upload circuits, equivalent to parity constrained generalized QSP blocks, have scalar entries that are arbitrary Laurent polynomials in \(z\), equivalently finite 4-periodic exponential polynomials in \(x\). This property is intrinsic to the rigid nature of \(\frac{\pi}{2}\)-frequencies of fixed upload gates. The approximation problem, however, is only imposed on the smaller interval \([0,1]\), or \([0,1]^m\) in multivariate case, as the rigidity constraints prohibit uniform approximation over the global domain. The upper bounds avoid the rigidity by restricting target approximation domain, then constructing a good 2-periodic auxiliary extension that allows Jackson theorem with Gevrey-\(\sigma\) class estimation, then realizing arbitrary generalized QSP circuits that have entries of Laurent polynomial over \(e^{i\pi x}\). The lower bounds exploit the same rigidity property in the opposite direction: if the target has a forced 4-periodic mismatch even outside the restricted approximation domain, then a finite mode fixed upload circuit cannot hide that mismatch obstruction too efficiently even in restricted local domain. Intuitively, our work is mainly focused on how the intrinsic structures of both the target and the hypothesis, even if they were outside the target interval of uniform approximation, affect the depth-error scaling.

\subsection{Auxiliary extensions and upper bounds}

The upper bound construction starts from the fact that the fixed upload circuits have intrinsic rigidity. Its scalar entries live naturally on the 4-periodic torus, because the basic fixed upload gate is \(e^{i\frac\pi2 x\sigma_z}\), where the \(\frac\pi2\)-frequencies lock the whole structure on the 4-torus. A target such as \(e^{i\eta x}\) is generally not continuous 4-periodic unless \(\eta\in\frac\pi2\mathbb Z\), so it cannot be approximated uniformly on the whole periodic domain by continuous 4-periodic functions with vanishing uniform error. Similarly, fixed upload circuits also have 2-\((-1)^D\)periodicity, meaning they have \(f(x)=(-1)^Df(x+2)\) over their domain. This is as restricting as the periodicity, if not more so. Therefore universal uniform approximation with arbitrary error \(\varepsilon\) is impossible over the global periodic domain, if the target class contains functions disagreeing with these two rigidity conditions. 

The way around this obstruction is to approximate only on the restricted local interval \([0,1]^m\), a subset ignoring the mismatch obstructions that occur on the boundaries and periodicity midpoints. Intuitively, it works similarly as using the injectivity of the sine wave over its first 1/4 period, thus avoiding the rigidity on half period and full period. 

Generalized QSP theorem connects each single variable fixed upload circuits with a GQSP unitary with Laurent polynomial having their upload frequencies \(\frac\pi2\), i.e. upload gates of \(e^{i\frac\pi2 x}\), they naturally have 4-periodicity and lives on 4-torus. But the rigidity of the fixed upload gate imposes parity constraint over these Laurent polynomials. In order to gain a dense function subset, we manipulate over the circuit construction and obtain realization of arbitrary Laurent polynomial over \(e^{i\pi x}\), they naturally live on the 2-torus. This is where the 2-periodicity comes from, they are the closest dense functions class that are realizable by fixed upload circuits' entries.

In order to utilize tools from classical approximation theory regarding polynomial approximation, where results are established on the whole torus, we have to re-extend the target to the global periodic domain in a meaningful way. Outside the restricted local interval one is free to choose an auxiliary function that agrees with target, for example \(e^{i\eta x}\) on \([0,1]\), but satisfies the global 2-periodic regularity needed for meeting the rigidity of generalized QSP block realizable by fixed upload circuits. This is the auxiliary extension mechanism behind our upper bounds. Intuitively, it is a scheme that tries to hide the defect of our hypothesis class by artificial manipulations. In~\cite{perezsalinas2025universal}, the universality of fixed upload circuits is proven exactly by this scheme. They chose a mirroring extension guaranteeing the regularity of the auxiliary extension up to continuity over 2-period, so that Fejér's theorem can apply.

In our upper bound theorems, auxiliary extension is chosen to be of Gevrey class. This choice guarantees the target's regularity up to infinite smoothness over 2-period and admits a bound for \(r\)-th derivative. Smoothness allows us to use Jackson's theorem over the torus for a sharp depth-error scaling. By Jackson's theorem, the more regular the extension is on the 2-periodic domain, the faster its Laurent polynomial approximants are to converge. Therefore, Gevrey class admits a satisfactory rate of this convergence, yielding the polylogarithmic depth-error scaling results. Generalized QSP theorem is essential, as it bridges the realizable Laurent polynomial approximants to feasible fixed upload circuit. In this sense, QSP discretizes the desired unitary behavior into repeated applications of a fixed upload, with the trainable freedom moved into the interleaved rotations. The single upload gate argument extends to the circuit wise upper bound by replacing each tunable upload gate separately and adding the resulting errors. This separable replacement is sufficient for the upper bound. 

Generalized QSP theorem has also been used to obtain logarithmic precision dependence for fractional powers of arbitrary unitary oracles\cite[Theorem~10]{motlagh2024generalized}. That result concerns a different computational model, allowing controlled oracle access, ancillary phase processing, and a spectral gap assumption. Its controlled query complexity does not constitute a depth bound for the ancilla free fixed upload architecture considered here. We instead use the polynomial completion component of GQSP as a synthesis theorem after constructing approximants satisfying the global periodicity constraints imposed by fixed uploads.

Aside from the theorems we proved, this general scheme works for similar problems in related fields, and in some sense it is the natural approximation scheme for general parametrized quantum circuit where the data are encoded through gates with fixed frequencies. The characteristics of these approximants make the tools along the line of auxiliary extension useful. Therefore, we see a possibility to further use the tools applied in this work in more general approximation problems with fixed upload parametrized quantum circuits.

\subsection{Mismatch obstruction and Turán-Nazarov inequality}

In the above mentioned auxiliary extension scheme for upper bound, an artificially crafted auxiliary extension allows us to say something about the upper bound, this can be understood as we are manually hiding the obstruction of the target function. From another perspective, it is also interesting to study how much this surgery to hide the obstruction costs, and whether we can extract some information from the obstructions even if it is hidden outside of our target interval. This observation enables us to derive a lower bound.

The analytic tool measuring this obstruction is the Turán-Nazarov inequality for exponential polynomials~\cite{nazarov1993local}. In the form relevant here, an exponential polynomial
\begin{equation}
    p(t)=\sum_{j=1}^{s}a_j e^{i\lambda_j t},
    \qquad \lambda_j\in\mathbb R,
\end{equation}
cannot be uniformly small on a set \(E\) and large on a surrounding interval \(I\) unless the number \(s\) of exponential modes is large:
\begin{equation}
    \sup_{t\in I}\abs{p(t)}
    \le
    \left(\frac{\Lambda_0\abs I}{\abs E}\right)^{s-1}
    \sup_{t\in E}\abs{p(t)} .
\end{equation}
Intuitively, local smallness propagates across the larger interval with a cost exponential in the number of modes. Thus, if some endpoint or off domain value is forced to remain large by some given condition, in our cases the global structure has given positive mismatch, then the exponential polynomial must contain at least logarithmically many modes. Further refinements and constants for this inequality are discussed in~\cite{FriedlandYomdin2013TuranNazarov}.

We proved that essentially all circuits we consider have the form of finite exponential polynomial of some set of frequencies. In fact, all finite quantum circuits with data upload gates have the form of finite exponential polynomials. For the tunable upload circuits this set is a finite subset of \(\mathbb{R}\), for the fixed upload circuits this set is a finite subset of \(\frac\pi2 \mathbb{Z}\). Therefore, the error can be expressed by a finite exponential polynomial, allowing us to estimate its behavior with Turán-Nazarov inequality. On the locality of \([0,1]^m\), we have the uniform approximation condition that says it is this small here. On the global feature over \([0,4]^m\), we can get a mismatch at the boundary, as long as the target is of mismatch class. Note that periodic mismatch simply records the inevitable error at the boundary of the torus. Turán-Nazarov turns any off interval mismatch into a lower bound on the number of modes, and hence into a lower bound on fixed upload depth. For fixed upload circuits, due to the rigidity, two types of mismatches can be extracted, one is the 4-periodic mismatch we used, another is the 2-\((-1)^D\)periodic mismatch. With careful definition, 2-\((-1)^D\)periodicity also provides a broad circuit subset, similar to that of 4-periodic mismatch class. However, as they are structurally similar and provide the same logarithmic depth-error growth rate, we did not explicitly talk about it, and it could be of future research interest.

We summarize this mechanism that produces lower bounds as mismatch obstruction. In fixed upload parametrized quantum circuits approximations, the obstruction exists whenever the target function or circuit violates the rigidity of the hypothesis circuits, while maintaining the exponential polynomial structure typical to finite parametrized quantum circuits. 

\section{Conclusion and Outlook}\label{sec:conclusion}
Fixed upload circuits provide a striking example of universality emerging from an extremely constrained architecture. Although recent work established that fixed upload circuits can approximate the same function classes as tunable upload circuits, the quantitative resource cost of this universality remained largely unexplored. In this work, we developed a depth-error scaling theory for fixed upload approximation and showed that the expressive power lost by removing tunable upload frequencies can be recovered with polylogarithmic growth and overhead. Beyond these quantitative bounds, our analysis reveals two structural mechanisms underlying approximation in fixed upload and QSP-type constructions. The first is an auxiliary extension mechanism, in which a target function restricted to target domain is re-extended into a globally admissible class compatible with the approximation architecture. The second is a mismatch obstruction mechanism, in which structural incompatibilities between the target and approximating family produce unavoidable depth-error lower bounds. While these techniques arise naturally in the fixed upload setting, we observe they may provide useful tools for studying approximation complexity more broadly in quantum signal processing and related quantum learning models.

Our study mainly focuses on the quantitative theory of approximation between fixed upload circuits and tunable upload circuits. The extension of our theorems to continuous function class suffer strong assumptions of appropriate conditions. For upper bound, it remains a theory between fixed upload circuits and tunable circuits, approximation of continuous function has to act through the proxy of tunable circuits, and therefore assuming strong assumptions. The question of whether the upper bounds can be carried to continuous function class under minimal assumptions is to be investigated further. For lower bound, the extension to continuous function class assumes the existence of a finite exponential polynomial representation of the target, as well as the existence of a mismatch obstruction. These assumptions are light in parametrized quantum circuits, but heavier in the realm of continuous function class. The function class for which these lower bounds are powerful is yet to be characterized. For both upper and lower bounds, the target dependent constants can be further analyzed.

From the perspective of quantum machine learning, our results clarify how expressivity is redistributed within restricted quantum architectures. Fixed upload circuits eliminate tunable upload frequencies while preserving universality, transferring the corresponding approximation burden into circuit depth. The fact that this transfer incurs only a near-logarithmic cost suggests that fixed upload circuits retain much of the expressive power of tunable upload circuits despite their substantially reduced upload flexibility. Several questions remain open. Most notably, closing the remaining gap between the logarithmic lower bound and the polylogarithmic upper bound would provide a complete characterization of the approximation complexity of fixed upload universality. It would also be interesting to determine whether stronger mismatch principles can yield sharper lower bounds, and whether the auxiliary extension framework developed here can be adapted to other QSP or Fourier-based quantum approximation schemes. We hope that the techniques introduced in this work provide a useful starting point for a broader quantitative theory of approximation complexity in quantum learning and quantum signal processing.

\vspace{0.5cm}

{\em Acknowledgments:} This research is supported by the National Research Foundation, Singapore through the National Quantum Office, hosted in A*STAR, under its Centre for Quantum Technologies Funding Initiative (S24Q2d0009). This research is also supported by A*STAR under its Young Investigator Research Grant (YIRG) M25N8c0131. We thank José Ignacio Latorre, Adrián Pérez-Salinas and Mahtab Yaghubi Rad for discussions and comments.

\clearpage

\section*{Appendices}
\begin{center}
\emph{Detailed proofs and supplementary technical material.}
\end{center}

\appendix

\section{Proofs for rigidity}
\phantomsection\label{proof:rigidity}
\begin{proof}[Proof of \cref{thm:rigidity}]
Without loss of generality, we reduce to a single variable fixed upload circuit \(U_{D}^f\). The multivariate case can be obtained by taking a coordinate slice composition. Write a fixed upload circuit with \(D\) data uploads as
\begin{equation}
    U_D^{f}(x)
    =
    B_De^{i\pi x\sigma_z/2}B_{D-1}e^{i\pi x\sigma_z/2}\cdots
    B_1e^{i\pi x\sigma_z/2}B_0 ,
\end{equation}
where each \(B_j\) is independent of \(x\). For every
coordinate \(k\),
\begin{equation}
    e^{i\pi 0\sigma_z/2}=I,\qquad
    e^{i\pi 2\sigma_z/2}=e^{i\pi\sigma_z}=-I,\qquad
    e^{i\pi 4\sigma_z/2}=e^{2i\pi\sigma_z}=I.
\end{equation}
These are all scalar multiples of the identity; therefore they commute through all trainable gates. Hence
\begin{equation}
\begin{aligned}
    U_D^{f}(0)
    &=B_DB_{D-1}\cdots B_0,\\
    U_D^{f}(2)
    &=B_D(-I)B_{D-1}\cdots B_1(-I)B_0
      =(-1)^D B_DB_{D-1}\cdots B_0,\\
    U_D^{f}(4)
    &=B_DB_{D-1}\cdots B_0 .
\end{aligned}
\end{equation}
Therefore
\begin{equation}
    U_D^{mf}(0)=U_D^{mf}(4)=(-1)^D U_D^{mf}(2).
\end{equation}

Consequently, let \(V\) be a target unitary valued function on a set\(E\) containing \([0,2]\), and suppose that \(U_D^{f}\) approximates \(V\) there with uniform operator norm error at most \(\eps\).
\begin{equation}
    \sup_{E}\norm{V(x)-U_D^{f}(x)}_{\mathrm{op}}
    \le
    \varepsilon
\end{equation}
Then inevitably,
\begin{equation}
\begin{aligned}
    \norm{V(2)-(-1)^DV(0)}_{\mathrm{op}}
    &\le
    \norm{V(2)-(-1)^DU_D^f(0)}_{\mathrm{op}}
    +
    \norm{(-1)^DU_D^f(0)-(-1)^DV(0)}_{\mathrm{op}}\\
    &=    
    \norm{V(2)-U_D^f(2)}_{\mathrm{op}}
    +
    \norm{U_D^f(0)-V(0)}_{\mathrm{op}}\\
    &\le
    2\varepsilon
\end{aligned}
\end{equation}
This means that
\begin{equation}
    min\{\norm{V(2)-V(0)}_{\mathrm{op}},\norm{V(2)+V(0)}_{\mathrm{op}}\}
    \le
    2\varepsilon
\end{equation}
Thus, we get a constant lower bound of the approximation error only dependent upon whether \(V\) has the same \(2\)-\((-1)^D\)periodicity that \(U_D^{mf}\) has. Without assuming about the nature of \(V\), we cannot get universal approximation of tunable upload circuits over \([0,2]\). 

An example would be the single residual tunable gate \(V(x)=e^{i\eta x \sigma_z}\) with \(\eta=\frac{\pi}{4} \). We would get \(\frac{\sqrt{2}}{2} \le \varepsilon\). Therefore, we can't uniformly approximate tunable circuits on \([0,2]\)
\end{proof}

\section{Proofs for upper bounds}
\label{app:upper-bounds}
\subsection{Single tunable upload gate upper bound}
We approximate one tunable upload gate over \([0,1]\) with a fixed upload circuit in this subsection. Notic that our target is
\begin{equation}
e^{iwx\sigma_z} = 
\begin{pmatrix}
    e^{iwx} & 0\\
    0 & e^{-iwx}
    \end{pmatrix},
    \quad 0 \le x \le 1.
\end{equation}
For every \(w\) we can write \(w = K\frac\pi2 + \eta\), separating \(w\) into a \(\frac\pi 2\) integer part and a residual part. The integer part can be implemented by simply adding \(|K|\) overhead gates. Therefore our only problem is to approximate the residual tunable upload gate \(e^{i\eta x \sigma_z}\) by fixed upload circuits.

Let \(r\in\mathbb N\) and \(\eta\in(0,\frac{\pi}{4})\). Call \(A\) a \(C_{2}^r[0,2]\) auxiliary extension function for \(e^{i\eta x}\) if
\begin{equation}
    A\in C_{2}^r[0,2],\qquad
    A(x+2)=A(x),\qquad
    \norm{A}_\infi\le1,
\end{equation}
and
\begin{equation}
    A(x)=e^{i\eta x},
    \qquad 0\le x\le1.
\end{equation}

Notice that an auxiliary extension function \(A\) is a \(r\)-smooth, 2-periodic extension of \(e^{i\eta x}\).

\begin{theorem}[Single residual tunable upload gate Jackson theorem]
\label{thm:single-residual}
Let \(A\) be a \(C_{2}^r[0,2]\) auxiliary extension for \(e^{i\eta x}\). For every \(N\ge1\), there exists a fixed upload circuit \(R_N^{(\eta,A)}(x)\), using \(O(N)\) copies of \(e^{i\frac\pi2 x\sigma_z}\), such that
\[
    \sup_{x\in[0,1]}
    \norm{R_N^{(\eta,A)}(x)-e^{i \eta x \sigma_z}}_F
    \le
    4\sqrt{\norm{A^{(r)}}_\infi}\,(2N+2)^{-r/2}.
\]
\end{theorem}

\begin{proof}
Apply Jackson IV theorem for 2-periodic \(r\)-smooth function \ref{thm:jackson} to \(A\). There exists
\begin{equation}
    T_N(x)=\sum_{n=-N}^{N}c_ne^{i\pi n x}
\end{equation}
such that
\begin{equation}
    \alpha_N
    \defeq
    \norm{A-T_N}_\infi
    \le
    \frac{\norm{A^{(r)}}_\infi}{(2N+2)^r}.
\end{equation}
This scalar polynomial \(T_N\) is not necessarily generalized-QSP realizable, because its sup norm may be larger than one. Since \(\norm{A}_\infi\le1\),
\begin{equation}
    \norm{T_N}_\infi
    \le
    \norm{A}_\infi+\norm{A-T_N}_\infi
    \le
    1+\alpha_N .
\end{equation}
We can define a normalized N-order Laurent polynomial by
\begin{equation}
    P_N(e^{i\pi x})
    \defeq
    \frac{T_N(x)}{1+\alpha_N}.
\end{equation}
Then
\(
    \sup_{\abs z=1}\abs{P_N(z)}\le1,
\)
so \(P_N\) is realizable with generalized QSP circuits. The normalization changes the scalar approximation by at most another \(\alpha_N\):
\begin{equation}
    \norm{A-P_N(e^{i\pi\cdot})}_\infi
    =
    \norm{A-\frac{T_N}{1+\alpha_N}}_\infi    
    \le
    \frac{\norm{A-T_N}_\infi+\alpha_N\norm{A}_\infi}{1+\alpha_N}
    \le
    2\alpha_N .
\end{equation}
Set
\begin{equation}
    \delta_N
    \defeq
    \norm{A-P_N(e^{i\pi\cdot})}_\infi .
\end{equation}
Then
\begin{equation}
    \delta_N
    \le
    \frac{2\norm{A^{(r)}}_\infi}{(2N+2)^r}.
\end{equation}

By Generalized QSP theorem~\ref{thm:gqsp-background}, \(P_N\) is the upper left entry of a unitary, obtained in the form of a generalized QSP circuit, which is realizable by a one dimensional fixed upload circuit with upload gate \(e^{i\frac\pi2x \sigma_z}\)
\begin{equation}
    R_N^{(\eta,A)}(x)
    =
    \begin{pmatrix}
    P_N(e^{i\pi x}) & -Q_N(e^{i\pi x})^*\\
    Q_N(e^{i\pi x}) & P_N(e^{i\pi x})^*
    \end{pmatrix},
\end{equation}
where
\(
    \abs{P_N(e^{i\pi x})}^2+\abs{Q_N(e^{i\pi x})}^2=1 .
\)

On \([0,1]\), the auxiliary extension function equals the target function, so
\begin{equation}
    \abs{P_N(e^{i\pi x})-e^{i\eta x}}\le\delta_N,
    \qquad 0\le x\le1.
\end{equation}
The whole block is controlled by calculating any operator norm, resulting an square root scale. Indeed, because \(\abs{P_N(e^{i\pi x})}\le1\),

\begin{equation}
    \abs{Q_N(e^{i\pi x})}^2
    =
    1-\abs{P_N(e^{i\pi x})}^2    
    \le
    2\bigl(1-\abs{P_N(e^{i\pi x})}\bigr)    
    \le
    2\abs{e^{i\eta x}-P_N(e^{i\pi x})}
    \le
    2\delta_N .
\end{equation}

Writing \(p_N(x)=P_N(e^{i\pi x})\) and \(q_N(x)=Q_N(e^{i\pi x})\), we get

\begin{equation}
    \norm{R_N^{(\eta,A)}(x)-E_\eta(x)}_F^2
    =
    2\abs{p_N(x)-e^{i\eta x}}^2
    +2\abs{q_N(x)}^2    
    \le
    2\delta_N^2+4\delta_N .
\end{equation}
If \(\alpha_N\le1/2\), then \(\delta_N\le1\) and hence
\begin{equation}
    \norm{R_N^{(\eta,A)}(x)-E_\eta(x)}_F
    \le
    \sqrt{6\delta_N}
    \le
    \sqrt{12\norm{A^{(r)}}_\infi}\,(2N+2)^{-r/2}.
\end{equation}
If \(\alpha_N>1/2\), then the right hand side of the theorem is at least
\begin{equation}
    4\sqrt{\alpha_N}>2\sqrt2,
\end{equation}
while the Frobenius distance between two \(2\times2\) unitaries is at most \(2\sqrt2\). Thus the same bound holds in all cases after replacing \(\sqrt{12}\) by \(4\).
\end{proof}

\begin{corollary}[Single tunable upload gate Jackson theorem]
\label{cor:Jackson-single}
Let \(w\in\mathbb R\), and choose
\begin{equation}
    K\in\mathbb Z,\qquad
    \eta=w-K\frac\pi2,\qquad
    \abs{\eta}\le \frac{\pi}{4}.
\end{equation}
Let \(A\) be a \(C_{2}^r[0,2]\) auxiliary extension for \(e^{i\eta x}\). Then the tunable upload gate \(e^{iw x\sigma_z}\) can be approximated uniformly on \([0,1]\) to Frobenius error at most \(\eps\) by a fixed upload circuit with upload depth
\begin{equation} 
    D=
    \abs K
    +
    O\!\left[
    \left(
    \frac{\norm{A^{(r)}}_\infi}{\eps^2}
    \right)^{1/r}
    \right].
\end{equation}
\end{corollary}

\begin{proof}
The split gives
\begin{equation}
    e^{iw x\sigma_z}
    =
    e^{iK\frac\pi2 x\sigma_z}e^{i\eta x\sigma_z}
\end{equation}
The factor \((e^{i\frac\pi2 x\sigma_z})^K\) is exact and costs \(\abs K\) fixed upload gates, counting \(\bigl(e^{i\frac\pi2 x\sigma_z}\bigr)^\dagger\) as one fixed upload gate if \(K<0\), as the inverse can be realized by adding \(X\) rotations as exhibited in \ref{fixeduploadcircuits}. Define
\begin{equation}
    U_D^{1f}(x)
    \defeq
    \bigl(e^{i\frac\pi2 x\sigma_z}\bigr)^K R_N^{(\eta,A)}(x),
\end{equation}
where \(R_N^{(\eta,A)}\) is the residual replacement from Theorem~\ref{thm:single-residual}. By unitary invariance of the Frobenius norm,

\begin{equation}
    \norm{U_D^{1f}(x)-e^{iw x\sigma_z}}_F
    =
    \norm{\bigl(e^{i\frac\pi2 x\sigma_z}\bigr)^K
    \bigl(R_N^{(\eta,A)}(x)-E_\eta(x)\bigr)}_F    
    =
    \norm{R_N^{(\eta,A)}(x)-E_\eta(x)}_F .
\end{equation}

Thus it is enough to choose \(N\) so that
\begin{equation}
    4\sqrt{\norm{A^{(r)}}_\infi}\,(2N+2)^{-r/2}
    \le
    \eps .
\end{equation}
Equivalently,
\begin{equation}
    N
    =
    O\!\left[
    \left(
    \frac{\norm{A^{(r)}}_\infi}{\eps^2}
    \right)^{1/r}
    \right].
\end{equation}
Adding the exact \(\abs K\) uploads proves the corollary.
\end{proof}

\subsection{Choosing the best auxiliary extension function class}
The single tunable upload gate theorem reduces the upper bound to controlling
\begin{equation}
    \delta_N
    \lesssim
    \frac{\norm{A^{(r)}}_\infi}{(2N+2)^r}.
\end{equation}
This is a Jackson type approximation argument. The error to \(N\) scaling is determined by the regularity, or explicitly, property of the derivatives of auxiliary extension function, which we can design at our convenience. In general, the more we can say about \(\norm{A^{(r)}}_\infi\), the more we know about this depth-error scaling. Therefore, the sharpest approximation upper bound that this model can potentially obtain depends on our choice of auxiliary extension function.

In the existing literature~\cite{perezsalinas2025universal}, an mirroring auxiliary extension is chosen, which satisfies regularity up to continuity over the torus. But this does not satisfy the regularity needed for Jackson's theorem over \(r\)-smoothness, for \(r>1\). Therefore, selecting a function class that is admissible and satisfies Jackson's requirement is essential for our purpose.

One possibility of auxiliary extension function that meets the r-smoothness and periodicity requirement is to use a finite order Hermite interpolation on the phase, by extending the \(\eta x\) function over\([0,1]\) to a Hermite polynomial interpolated composite function over \([0,2]\). This gives a \(C_{2}^r[0,2]\) auxiliary extension function, but by itself it does not admit any control of \(\norm{A^{(r)}}_\infi\) as \(r\) varies. In general, an arbitrary function of mere smooth class admits no bound over the \(r\)-th derivative. These are exactly the finite smoothness methods and admits polynomial dependence on \(\frac1\varepsilon\).

Analytic functions would be ideal auxiliary extension functions, because Cauchy estimates control all derivatives, but they are unavailable for non-integer residual frequency: if a 2-periodic analytic function agrees with \(e^{i\eta x}\) on an open interval, analytic continuation forces it to be \(e^{i\eta x}\) everywhere, and 2-periodicity then requires \(e^{2i\eta}=1\), i.e. \(\eta\in\pi\mathbb Z\subset \frac\pi2 \mathbb{Z}\). This gives a contradiction.

Therefore, for our purpose, we choose something in between analytic class and smooth class. The natural choice is the 2-periodic Gevrey class \(G^\sigma_2\) from Definition~\(\ref{def:2-periodic-gevrey}\). Indeed, this class allows a bound to \(r\)-th derivative, and we will show in later sections that there exist such auxiliary extensions that satisfies this regularity condition. But first, we examine the improvement of depth-error bound for choosing Gevrey class auxiliary extensions.

\begin{lemma}[Gevrey class Jackson type approximation]
\label{lem:gevrey-jackson}
Let \(A\in G^\sigma_2\), \(\sigma>1\), then there exist constants \(C_{\sigma,A},c_{\sigma,A}>0\) such that, for every \(N\ge1\), there is \(T_N\in\mathcal T_N^{(2)}\) satisfying
\begin{equation}
    \norm{A-T_N}_\infi
    \le
    C_{\sigma}e^{-c_{\sigma}N^{1/\sigma}}.
\end{equation}
\end{lemma}

\begin{proof}
Apply Jackson's IV theorem for 2-torus \ref{thm:jackson} with derivative order \(k\). The Gevrey bound gives
\begin{equation}
    \norm{A-T_N}_\infi
    \le
    C\frac{R^k(k!)^\sigma}{(2N+2)^k}.
\end{equation}
Using \(k!\le k^k\),
\begin{equation}
    \norm{A-T_N}_\infi
    \le
    C\left(\frac{Rk^\sigma}{2N+2}\right)^k .
\end{equation}
Choose
\begin{equation}
    k=\left\lfloor a_\sigma N^{1/\sigma}\right\rfloor
\end{equation}
with \(a_\sigma>0\) small enough that \(Ra_\sigma^\sigma\le e^{-2}\). Then, for all sufficiently large \(N\),
\begin{equation}
    \frac{Rk^\sigma}{2N+2}\le e^{-2},
\end{equation}
and the error is at most \(Ce^{-2k}\). Since \(k\ge (a_\sigma/2)N^{1/\sigma}\) for all sufficiently large \(N\), this gives the claimed stretched-exponential estimate after adjusting constants for the finitely many small \(N\).
\end{proof}

\begin{remark}
  Gevrey class admits sub-exponential uniform convergence for polynomial approximation on torus.
\end{remark}

\subsection{Constructing a Gevrey auxiliary extension function}

In this subsection, we construct a Gevrey class auxiliary extension function by several layers of composition and verify that these compositions preserves the Gevrey property.

For \(\sigma>1\), set
\begin{equation}
    \alpha=\frac{1}{\sigma-1}.
\end{equation}
Then \(\sigma=1+1/\alpha\). Define
\begin{equation}
    \rho_\sigma(t)
    =
    \begin{cases}
    \exp(-t^{-\alpha}), & t>0,\\
    0, & t\le0,
    \end{cases}
\end{equation}
and
\begin{equation}
    \chi_\sigma(t)
    =
    \frac{\rho_\sigma(t)}
    {\rho_\sigma(t)+\rho_\sigma(1-t)},
    \qquad 0\le t\le1.
\end{equation}

\begin{lemma}[\(\rho_\sigma\) and \(\chi_\sigma\) are in \(G^\sigma\)]
\label{lem:flat-gevrey-smoothstep} The function \(\rho_\sigma\) belongs to \(G^\sigma\) on every compact interval. Moreover, \(\chi_\sigma\in G^\sigma([0,1])\),
\begin{equation}
    \chi_\sigma(0)=0,\qquad \chi_\sigma(1)=1,
\end{equation}
and
\begin{equation}
    \chi_\sigma^{(n)}(0)=\chi_\sigma^{(n)}(1)=0,
    \qquad n\ge1.
\end{equation}
Consequently, the extension
\begin{equation}
    \chi_\sigma(t)=0 \quad (t\le0),
    \qquad
    \chi_\sigma(t)=1 \quad (t\ge1)
\end{equation}
is Gevrey-\(\sigma\) on compact intervals.
\end{lemma}

\begin{proof}
Away from \(t=0\), \(t\mapsto e^{-t^{-\alpha}}\) is analytic, hence Gevrey of every order. It remains to control the derivatives as \(t\downarrow0\). For \(0<t\le1\), use the holomorphic extension \(z\mapsto e^{-z^{-\alpha}}\) in a disk \(\abs{z-t}\le\delta t\), with \(\delta>0\) chosen small enough that a single branch is used and
\begin{equation}
    \operatorname{Re}(z^{-\alpha})\ge c t^{-\alpha}
\end{equation}
for some \(c>0\). Cauchy's estimate gives
\begin{equation}
    \abs{\rho_\sigma^{(n)}(t)}
    \le
    n!(\delta t)^{-n}e^{-c t^{-\alpha}}
    =
    n!\delta^{-n}t^{-n}e^{-c t^{-\alpha}} .
\end{equation}
Set \(u=t^{-\alpha}\). Then
\begin{equation}
    t^{-n}e^{-c t^{-\alpha}}
    =
    u^{n/\alpha}e^{-cu}.
\end{equation}
The maximum over \(u>0\) occurs at \(u=n/(\alpha c)\), so
\begin{equation}
    u^{n/\alpha}e^{-cu}
    \le
    \left(\frac{n}{\alpha c e}\right)^{n/\alpha}.
\end{equation}
Therefore
\begin{equation}
    \abs{\rho_\sigma^{(n)}(t)}
    \le
    n!\delta^{-n}
    \left(\frac{n}{\alpha c e}\right)^{n/\alpha}.
\end{equation}
Using \(n^n\le e^n n!\), we get
\begin{equation}
    n^{n/\alpha}\le e^{n/\alpha}(n!)^{1/\alpha},
\end{equation}
and hence
\begin{equation}
    \abs{\rho_\sigma^{(n)}(t)}
    \le
    C L^n n!(n!)^{1/\alpha}
    =
    C L^n (n!)^{1+1/\alpha}
    =
    C L^n (n!)^\sigma .
\end{equation}
After increasing \(C,L\), the same estimate holds on any compact interval. Since \(e^{-t^{-\alpha}}\) decays faster than any power of \(t\), every derivative tends to \(0\) at \(t=0^+\). Thus
\begin{equation}
    \rho_\sigma^{(n)}(0)=0,\qquad n\ge0,
\end{equation}
and \(\rho_\sigma\) is flat at the endpoint.

The translated \(t\mapsto\rho_\sigma(1-t)\) is also Gevrey-\(\sigma\) for free. The denominator
\begin{equation}
    D(t)=\rho_\sigma(t)+\rho_\sigma(1-t)
\end{equation}
is strictly positive on \([0,1]\): near \(0\), \(\rho_\sigma(1-t)>0\), while near \(1\), \(\rho_\sigma(t)>0\). Since Gevrey classes with \(\sigma>1\) are closed under multiplication and under taking inverse if nonvanishing~\cite{RainerSchindl2014Composition}, \(D^{-1}\in G^\sigma([0,1])\). Therefore
\begin{equation}
    \chi_\sigma(t)=\rho_\sigma(t)D(t)^{-1}
\end{equation}
is Gevrey-\(\sigma\) on \([0,1]\). The endpoint values are immediate, and the quotient is flat at \(0\) because \(\rho_\sigma(t)\) is flat there. Similarly,
\begin{equation}
    1-\chi_\sigma(t)
    =
    \frac{\rho_\sigma(1-t)}
    {\rho_\sigma(t)+\rho_\sigma(1-t)}
\end{equation}
is flat at \(1\). Hence all positive order derivatives of \(\chi_\sigma\) vanish at both endpoints, and the constant extension remains Gevrey-\(\sigma\).
\end{proof}

\begin{lemma}[Gevrey auxiliary extension \(A_{\eta,\sigma}\)]
\label{lem:gevrey-closure}
Let \(\sigma>1\) and \(\abs{\eta}\le\pi/4\). Define
\begin{equation}
    A_{\eta,\sigma}(x)
    \defeq
    e^{i\eta x}
    \exp\!\left(-2i\eta\,\chi_\sigma(x-1)\right)
    =
    \begin{cases}
    e^{i\eta x} \qquad 0\le x \le 1\\
    e^{i\eta x} \exp\!\left(-2i\eta\,\chi_\sigma(x-1)\right) \qquad 0\le x\le2
    \end{cases}
\end{equation}
extended periodically with 2-periodicity. Then \(A_{\eta,\sigma}\in G^\sigma_2\),
\begin{equation}
    \abs{A_{\eta,\sigma}(x)}=1,
    \qquad
    A_{\eta,\sigma}(x)=e^{i\eta x}\quad(0\le x\le1),
\end{equation}
and there exist constants \(C_\sigma,R_\sigma>0\), depending only on \(\sigma\), such that
\begin{equation}
    \norm{A_{\eta,\sigma}^{(k)}}_\infi
    \le
    C_\sigma R_\sigma^k(k!)^\sigma,
    \qquad k \in \mathbb{N} \cup \{0\}
\end{equation}
uniformly for all \(\abs{\eta}\le\pi/4\).
\end{lemma}

\begin{proof}
By Lemma~\ref{lem:flat-gevrey-smoothstep}, the extended \(\chi_\sigma(x-1)\) is Gevrey-\(\sigma\) on \([0,2]\). In
\begin{equation}
    A_{\eta,\sigma}(x)
    =
    e^{i\eta x}
    \exp(-2i\eta\chi_\sigma(x-1)).
\end{equation}
The factor \(x\mapsto e^{i\eta x}\) is analytic, hence Gevrey-\(1\), and therefore Gevrey-\(\sigma\). Since Gevrey classes with \(\sigma>1\) are closed under analytic composition and multiplication~\cite{RainerSchindl2014Composition}, \(A_{\eta,\sigma}\in G^\sigma([0,2])\).

At \(x=1\), \(\chi_\sigma(0)=0\), so the two branches agree:
\begin{equation}
    e^{i\eta}
    =
    e^{i\eta}\exp(-2i\eta\chi_\sigma(0)).
\end{equation}
Since \(\chi_\sigma^{(s)}(0)=0\) for all \(s\ge1\), the factor \(\exp(-2i\eta\chi_\sigma(x-1))\) has value \(1\) at \(x=1\) and all its positive order derivatives vanish there. Therefore the derivatives of the right branch at \(x=1\) are exactly those of \(e^{i\eta x}\).

The same flatness gives the periodic matching. At \(x=0\),
\begin{equation}
    A_{\eta,\sigma}^{(n)}(0)=(i\eta)^n .
\end{equation}
At \(x=2\), \(\chi_\sigma(1)=1\), so
\begin{equation}
    A_{\eta,\sigma}(2)
    =
    e^{2i\eta}e^{-2i\eta}
    =
    1
    =
    A_{\eta,\sigma}(0).
\end{equation}
For derivatives, \(\chi_\sigma^{(s)}(1)=0\) for all \(s\ge1\), so that only the derivatives of \(e^{i\eta x}\) contribute:
\begin{equation}
    A_{\eta,\sigma}^{(n)}(2)
    =
    (i\eta)^n e^{2i\eta}e^{-2i\eta}
    =
    (i\eta)^n
    =
    A_{\eta,\sigma}^{(n)}(0).
\end{equation}
Thus \(A_{\eta,\sigma}\) defines a Gevrey-\(\sigma\) function on the periodic interval \(\mathbb R/2\mathbb Z\).

The constants can be chosen uniformly for \(\abs{\eta}\le\pi/2\), because the Gevrey seminorms of the phase \(\eta x-2\eta\chi_\sigma(x-1)\) are bounded uniformly on this compact \(\eta\)-interval, and the analytic composition estimates are uniform on the corresponding bounded range. Therefore
\begin{equation}
    \norm{A_{\eta,\sigma}^{(k)}}_\infi
    \le
    C_\sigma R_\sigma^k(k!)^\sigma
\end{equation}
with constants depending only on \(\sigma\). The identity on \([0,1]\) follows from \(\chi_\sigma(x-1)=0\) there, and unit modulus is immediate.
\end{proof}

\phantomsection\label{proof:main-gevrey-upper}
\begin{proof}[Proof of \Cref{thm:main-gevrey-upper}]
Choose Gevrey class auxiliary extension in Lemma~\ref{lem:gevrey-closure} for the circuit realization, apply \cref{thm:single-residual}, apply Corollary~\ref{cor:Jackson-single} to estimate the depth-error scaling and use Lemma~\ref{lem:gevrey-jackson} for the sharpened rate.
\end{proof}

\subsection{Circuit wise upper bound}

\phantomsection\label{proof:main-circuit-upper}
\begin{proof}[Proof of \Cref{thm:main-circuit-upper}]
The Frobenius bound of \Cref{thm:main-gevrey-upper} implies the same bound in operator norm for each tunable upload gate. Write the tunable upload circuit and the fixed upload circuit as identical products except that each \(e^{iw_{jk}x_k\sigma_z}\) is replaced by a fixed upload block \(R_{K_{jk},N}^{(w_{jk})}(x_k)\). The telescoping identity for products gives
\begin{equation}
    \norm{U_L^w(x)-U_D^{mf}(x)}_{\mathrm{op}}
    \le
    \sum_{j=1}^{L}\sum_{k=1}^{m}
    \norm{
    e^{iw_{jk}x_k\sigma_z}
    -
    R_{K_{jk},N}^{(w_{jk})}(x_k)
    }_{\mathrm{op}},
\end{equation}
because all surrounding factors are unitary. Taking the supremum over \(x\) and applying \Cref{thm:main-gevrey-upper} proves the error estimate. The upload count is the sum of the \(O(\abs{K_{jk}}+N)\) costs of the \(L\,m\) coordinate blocks, and the displayed choice of \(N\) is obtained by solving
\begin{equation}
    L\,m\, C_\sigma e^{-c_\sigma N^{1/\sigma}}\le\eps .
\end{equation}
\end{proof}

\phantomsection\label{proof:dense-circuit-upper}
\phantomsection\label{proof:measured-circuit-upper}
\begin{proof}[Proof of \Cref{cor:measured-circuit-upper}]
Apply \(\Cref{thm:main-circuit-upper}\) with circuit error \(\eps\). This gives a coordinate-wise fixed upload circuit \(U_D^{mf}\) satisfying
\begin{equation}
    \sup_x\norm{U_L^w(x)-U_D^{mf}(x)}_{\mathrm{op}}\le\eps
\end{equation}
with fixed upload count
\begin{equation}
    O_\sigma\!\left(
    \sum_{j=1}^{L}\sum_{k=1}^{m}\abs{K_{jk}}
    +
    L\,m\left(\log\frac{L\,m\, C_\sigma}{\eps}\right)^\sigma
    \right).
\end{equation}
The construction in Ref.~\cite[Theorem~7]{perezsalinas2025universal} realizes coordinate fixed upload blocks inside the dense architecture generated by \(V_m(x)\), up to a constant overhead in dense fixed upload calls and \(x\)-independent trainable \(SU(2m)\) unitaries. Absorbing this constant overhead into the big-\(O_\sigma\) notation gives a dense fixed upload circuit \(U_D^{\mathrm{dense}}\) with the displayed dense upload count and the same active block action as \(U_D^{mf}\).

For scalar amplitudes,
\begin{equation}
    \abs{h_U(x)-h_V(x)}
    =
    \abs{\langle0|(U(x)-V(x))|0\rangle}
    \le
    \norm{U(x)-V(x)}_{\mathrm{op}} .
\end{equation}
Applying this to the active block gives
\begin{equation}
    \norm{h-h_{U_D^{\mathrm{dense}}}}_\infi\le\eps ,
\end{equation}
as claimed.
\end{proof}

\section{Proofs for lower bounds}
\label{app:lower-bounds}
\subsection{Single residual tunable upload gate lower bound}

Here we first prove the depth-error lower bound for approximating a single residual tunable upload gate with fixed upload circuits.

For delivering the upper bound for a single tunable upload gate we used two structural ingredients. The first is the restriction of the target domain and 2-periodic auxiliary extension that fits the regularity needed for polynomial approximation, and the second is generalized QSP polynomial form natural to the one variable fixed upload circuit structure. We cannot obtain a Gevrey class target function and approximate using Jackson theorem without playing the tricks over the domain; we cannot get a realization of fixed upload circuits without using generalized QSP theorem to bridge between polynomials and quantum circuits. 

Therefore, if we were to seek for a lower bound for the single residual tunable upload gate approximation, we should ask what these structures give us. Our single residual tunable upload gate lower bound exploits these ingredients but in the opposite direction. The tool we use is the Turán-Nazarov inequality for exponential polynomials, stated in \(\Cref{thm:turan-nazarov}\). This inequality fits our purpose because it relates the size of an exponential polynomial on a subinterval to its size on a larger interval.

\phantomsection\label{proof:main-single-gate-lower}
\begin{proof}[Proof of \Cref{thm:main-single-gate-lower}]
A single variable fixed upload circuit with \(N\) fixed upload gates coincides with generalized QSP circuits of Laurent degree at most \(N\), with uniform parity over modes, thus has the form
\begin{equation}
    U_N^{1f}
    =
    R_N(x)
    =
    \begin{pmatrix}
    P_N(e^{i\frac\pi2 x}) & -Q_N(e^{i\frac\pi2 x})^*\\
    Q_N(e^{i\frac\pi2 x}) & P_N(e^{i\frac\pi2 x})^*
    \end{pmatrix},
\end{equation}
with
\begin{equation}
    P_N(z)=\sum_{n=-N}^{N}p_nz^n,
    \qquad
    \abs{P_N(e^{i\frac\pi2 x})}^2+\abs{Q_N(e^{i\frac\pi2 x})}^2=1,
\end{equation}
and write
\begin{equation}
    p_N(x)=P_N(e^{i\frac\pi2 x}),
    \qquad
    q_N(x)=Q_N(e^{i\frac\pi2 x}),
    \qquad
    f_\eta(x)=e^{i\eta x}.
\end{equation}
The generalized QSP constraint gives
\begin{equation}
    \abs{p_N(x)}^2+\abs{q_N(x)}^2=1.
\end{equation}
Therefore

\begin{equation}
    \norm{R_N(x)-e^{i\eta x\sigma_z}}_F^2
    =
    2\abs{p_N(x)-f_\eta(x)}^2
    +
    2\abs{q_N(x)}^2    
    =
    4\left(
    1-\operatorname{Re}\bigl[p_N(x)e^{-i\eta x}\bigr]
    \right).
\end{equation}

Define
\begin{equation}
    g_N(x)
    \defeq
    1-\operatorname{Re}\bigl[p_N(x)e^{-i\eta x}\bigr].
\end{equation}
Since \(\abs{p_N(x)}\le1\), one has \(g_N(x)\ge0\). The error assumption implies
\begin{equation}
    0\le g_N(x)\le\frac{\eps^2}{4},
    \qquad 0\le x\le1.
\end{equation}

The fixed upload Laurent polynomial form is 4-periodic:
\begin{equation}
    p_N(4)=p_N(0).
\end{equation}
At \(x=0\), the Frobenius error bound gives
\begin{equation}
    \abs{p_N(0)-1}\le\frac{\eps}{\sqrt2}.
\end{equation}
Consequently,

\begin{equation}
    g_N(4)
    =
    1-\operatorname{Re}\bigl[p_N(0)e^{-4i\eta}\bigr]    
    \ge
    1-\cos(4\eta)-\abs{p_N(0)-1}    
    \ge
    2\sin^2 2\eta-\frac{\eps}{\sqrt2}.
\end{equation}
Under the stated small error condition,
\begin{equation}
    g_N(4)\ge\sin^2 2\eta .
\end{equation}

Now \(g_N\) is an exponential polynomial. Write
\begin{equation}
    p_N(x)=\sum_{n=-N}^{N}c_ne^{i\frac\pi2 nx},
\end{equation}
then
\begin{equation}
    g_N(x)
    =
    1-\frac12\sum_{n=-N}^{N}
    \left(
    c_ne^{i(\frac\pi2 n-\eta)x}
    +
    \overline{c_n}e^{-i(\frac\pi2 n-\eta)x}
    \right).
\end{equation}
Thus \(g_N\) has at most \(4N+3\) exponential terms. Apply Theorem~\ref{thm:turan-nazarov} with \(I=[0,4]\) and \(E=[0,1]\), and absorb the fixed ratio \(\abs I/\abs E=4\) into a universal constant \(\Lambda>1\). Then
\begin{equation}
    \sup_{x\in[0,4]}\abs{g_N(x)}
    \le
    \Lambda^{4N+2}
    \sup_{x\in[0,1]}\abs{g_N(x)}
    \le
    \Lambda^{4N+2}\frac{\eps^2}{4}.
\end{equation}
Since \(g_N(4)\ge\sin^2 2\eta\),
\begin{equation}
    \sin^2 2\eta
    \le
    \Lambda^{4N+2}\frac{\eps^2}{4}.
\end{equation}
Taking square roots gives
\begin{equation}
    \eps
    \ge
    2\abs{\sin 2\eta}\,\Lambda^{-(2N+1)} .
\end{equation}
Rearranging proves the logarithmic lower bound.
\end{proof}

\begin{remark}[Near optimality for single residual tunable upload gate approximation]
The Gevrey construction gives, for every fixed \(\sigma>1\),
\begin{equation}
    N
    =
    O_\sigma\!\left[
    \left(\log\frac1\eps\right)^\sigma
    \right],
\end{equation}
while \Cref{thm:main-single-gate-lower} gives
\begin{equation}
    N=\Omega_\eta\!\left(\log\frac1\eps\right)
\end{equation}
for every fixed non-integer residual frequency \(\eta\). Thus the construction is near-logarithmic in fixed upload depth. The remaining gap is the exponent \(\sigma>1\) and the behavior of the constants as \(\sigma\downarrow1\).
\end{remark}

This tells us that if a tunable upload circuit approximation is separated into single tunable upload gate replacements, we can get a near-logarithmic approximation, but each non-integer residual tunable upload gate approximation is never faster than logarithmic with respect to \(\frac{1}{\varepsilon}\).

\subsection{Circuit wise lower bound}
In the upper bound theorems, once we have an upper bound for single tunable upload gate approximation, we can obtain an upper bound for the circuit wise approximation by a worst case argument. We can replace each tunable upload gate separately and telescope the product error. This separability is not true in general for all possible approximations. There is no general requirement that a circuit wise approximation decomposes into separated single upload gate approximations. Interactions among tunable upload gates may make the product easier to approximate. For this reason, a circuit wise lower bound cannot be obtained by simply adding single residual tunable upload gate lower bounds. In fact, as we will prove later, there is no free lower bound that holds for all target tunable upload circuits without an additional obstruction.

However, we can generalize the approach used to establish the lower bound for a single residual tunable upload gate and investigate whether similar reasoning can be extended to derive a circuit level lower bound. Care is required in this generalization, as the argument applies only to the class of mismatch target functions. The main idea is that the domain asked for approximation is a subset of the domain of these approximating functions, and the rigid 4-periodicity condition of the fixed upload circuits allow us to utilize the mismatch, together with Turán-Nazarov inequality, this generates a lower bound.

\begin{lemma}[Exponential mode count for finite quantum circuits]
\label{lem:finite-phase-term-count}
Let
\begin{equation}
    U(t)
    =
    B_0
    \left(
    \prod_{\ell=1}^{L} e^{i\omega_\ell t H_\ell}B_\ell
    \right)
\end{equation}
be a finite dimensional matrix valued function, where the matrices \(B_0,\ldots,B_L\) are independent of \(t\), the numbers \(\omega_\ell\in\mathbb R\), and each \(H_\ell\) is Hermitian with finite spectrum \(\operatorname{spec}(H_\ell)\). Then every scalar matrix element \(\langle u|U(t)|v\rangle\) is an exponential polynomial whose frequencies are contained in
\begin{equation}
    \left\{
    \sum_{\ell=1}^{L}\omega_\ell\lambda_\ell:
    \lambda_\ell\in\operatorname{spec}(H_\ell)
    \right\}.
\end{equation}
In particular, each scalar matrix element has at most
\begin{equation}
    S_U
    \le
    \prod_{\ell=1}^{L}
    \abs{\operatorname{spec}(H_\ell)}
\end{equation}
distinct exponential terms. For one-qubit phase uploads
\(e^{iw_\ell t\sigma_z}\), this gives \(S_U\le 2^L\).
\end{lemma}

\begin{proof}
For each \(\ell\), write the spectral decomposition
\begin{equation}
    H_\ell
    =
    \sum_{\lambda\in\operatorname{spec}(H_\ell)}
    \lambda\,\Pi_{\ell,\lambda}.
\end{equation}
Then
\begin{equation}
    e^{i\omega_\ell tH_\ell}
    =
    \sum_{\lambda\in\operatorname{spec}(H_\ell)}
    e^{i\omega_\ell\lambda t}\Pi_{\ell,\lambda}.
\end{equation}
Substituting these decompositions into the product for \(U(t)\) gives
\begin{equation}
    U(t)
    =
    \sum_{\lambda_1,\ldots,\lambda_L}
    e^{i(\omega_1\lambda_1+\cdots+\omega_L\lambda_L)t}
    B_0\Pi_{1,\lambda_1}B_1\cdots
    \Pi_{L,\lambda_L}B_L .
\end{equation}
Taking the scalar matrix element between fixed vectors \(u,v\) preserves this finite expansion and only replaces each matrix coefficient by a scalar coefficient. Terms with the same frequency may be combined; therefore the number of distinct exponential terms is bounded by the number of spectral choices, \(\prod_{\ell=1}^{L}\abs{\operatorname{spec}(H_\ell)}\). Since \(\operatorname{spec}(\sigma_z)=\{-1,1\}\), the one qubit data upload circuit case has at most \(2^L\) terms.
\end{proof}

\phantomsection\label{proof:main-circuit wise-lower}
\begin{proof}[Proof of \Cref{thm:main-circuit wise-lower}]
Write
\begin{equation}
    \Delta=\Delta_4(U_L^w)
    =
    \norm{U_L^w(4)-U_L^w(0)}_{\mathrm{op}} .
\end{equation}
The operator norm can be tested by matrix elements. Hence there exist unit vectors \(u,v\in\mathbb C^d\) such that
\begin{equation}
    \abs{
    \langle u|
    \bigl(U_L^w(0)-U_L^w(4)\bigr)
    |v\rangle
    }
    =
    \Delta .
\end{equation}

Now look at the scalar error in exactly this matrix direction:
\begin{equation}
    h(t)
    =
    \langle u|
    \bigl(U_D^{1f}(t)-U_L^w(t)\bigr)
    |v\rangle .
\end{equation}
On the approximation interval \([0,1]\),
\begin{equation}
    \abs{h(t)}
    \le
    \norm{U_D^{1f}(t)-U_L^w(t)}_{\mathrm{op}}
    \le
    \eps .
\end{equation}
Therefore
\begin{equation}
    \sup_{t\in[0,1]}\abs{h(t)}\le\eps .
\end{equation}

The fixed upload circuit is 4-periodic entry wise. Thus
\begin{equation}
    U_D^{1f}(4)=U_D^{1f}(0).
\end{equation}
At \(t=4\), the scalar error is therefore forced to see the endpoint mismatch of the target:

\begin{equation}
    \begin{aligned}  
    \abs{h(4)}
    &=
    \abs{
    \langle u|
    \bigl(U_D^{1f}(4)-U_L^w(4)\bigr)
    |v\rangle
    } \\
    &=
    \abs{
    \langle u|
    \bigl(U_D^{1f}(0)-U_L^w(4)\bigr)
    |v\rangle
    } \\
    &\ge
    \abs{
    \langle u|
    \bigl(U_L^w(0)-U_L^w(4)\bigr)
    |v\rangle
    }
    -
    \abs{
    \langle u|
    \bigl(U_D^{1f}(0)-U_L^w(0)\bigr)
    |v\rangle
    } \\
    &\ge
    \Delta-\eps .
    \end{aligned}
\end{equation}

Since \(\eps\le\Delta/2\), this gives
\begin{equation}
    \abs{h(4)}\ge\frac{\Delta}{2}.
\end{equation}

Now count exponential terms. The fixed upload scalar \(\langle u|U_D^{1f}(t)|v\rangle\) has at most \(2D+1\) 4-periodic terms:
\begin{equation}
    e^{i\frac\pi2 nt},\qquad -D\le n\le D.
\end{equation}

By assumption, the target scalar \(\langle u|U_L^w(t)|v\rangle\) has at most \(S\) exponential terms. Therefore \(h(t)\) is an exponential polynomial with at most\((2D+1)+S\) terms.

Apply Theorem~\ref{thm:turan-nazarov} to \(h\) with
\(
    I=[0,4],
    E=[0,1].
\)
Since \(\abs I/\abs E=4\), we may absorb this fixed ratio into a universal constant \(\Lambda>1\). Because \(h\) has at most \(2D+1+S\) terms,
\begin{equation}
    \sup_{t\in[0,4]}\abs{h(t)}
    \le
    \Lambda^{2D+S}
    \sup_{t\in[0,1]}\abs{h(t)}
    \le
    \Lambda^{2D+S}\eps .
\end{equation}
But \(\abs{h(4)}\le\sup_{[0,4]}\abs{h}\) and \(\abs{h(4)}\ge\Delta/2\). Hence
\begin{equation}
    \frac{\Delta}{2}
    \le
    \Lambda^{2D+S}\eps .
\end{equation}
Equivalently,
\begin{equation}
    \log\!\left(\frac{\Delta}{2\eps}\right)
    \le
    (2D+S)\log\Lambda .
\end{equation}
Rearranging proves
\begin{equation}
    D
    \ge
    \frac{1}{2\log\Lambda}
    \log\!\left(\frac{\Delta}{2\eps}\right)
    -
    \frac{S}{2}.
\end{equation}

\begin{remark}
This proof is a reasonable generalization of the single residual gate case.
\end{remark}
\end{proof}

\subsection{No universal lower bound}

The theorem above applies to any tunable upload circuit with a mismatch obstruction, but the bound depends on the mismatch obstruction \(\Delta_4(U_L^w)\). One cannot remove this dependence.

\phantomsection\label{proof:main-no-mismatch-free-bound}
\begin{proof}[Proof of \Cref{thm:main-no-mismatch-free-bound}]
Take two non-integer residual tunable upload gates with opposite signs:
\begin{equation}
    U_L^w(t)
    =
    e^{i\eta t\sigma_z}
    e^{-i\eta t\sigma_z},
    \qquad
    \eta\notin \frac{\pi}{4}\mathbb Z .
\end{equation}
The whole product is exactly
\begin{equation}
    U_L^w(t)=I.
\end{equation}
It is therefore approximated with zero error by the depth-zero fixed upload circuit
\begin{equation}
    U_0^{1f}(t)=I.
\end{equation}
So even though the individual gates contain non-integer residual frequencies, the final circuit has no approximation hardness at all. This proves that a valid circuit wise lower bound must depend on a circuit wise obstruction such as \(\Delta_4(U_L^w)\), or else it must impose an additional no cancellation or isolation assumption.
\end{proof}

\section{Beyond: Approximating multivariate functions} \label{app:function-approx}
This section proves several function level statements in the main results section. It should be noted that conditions on the target functions are assumed to be narrower than arbitrary continuous function. In some sense, these are the most generalized versions of the theory we observed targeting functions with specific restrictions.

\phantomsection\label{proof:continuous-function-upper}
\begin{proof}[Proof of \Cref{cor:continuous-function-upper}]
First apply \(\Cref{thm:main-circuit-upper}\) with circuit error \(\eps/2\). This gives a coordinate-wise fixed upload circuit \(U_D^{mf}\) with
\begin{equation}
    \sup_x
    \norm{U_L^w(x)-U_D^{mf}(x)}_{\mathrm{op}}
    \le
    \frac{\eps}{2}
\end{equation}
and fixed upload count
\begin{equation}
    O_\sigma\!\left(
    \sum_{j=1}^{L}\sum_{k=1}^{m}\abs{K_{jk}}
    +
    L\,m\left(\log\frac{2L\,m\, C_\sigma}{\eps}\right)^\sigma
    \right).
\end{equation}
Since
\begin{equation}
    \abs{h_U(x)-h_V(x)}
    =
    \abs{\langle0|(U(x)-V(x))|0\rangle}
    \le
    \norm{U(x)-V(x)}_{\mathrm{op}},
\end{equation}
we have
\begin{equation}
    \norm{h-h_{U_D^{mf}}}_\infi
    \le
    \frac{\eps}{2}.
\end{equation}
Combining this with \(\norm{f-h}_\infi\le\eps/2\) gives
\begin{equation}
    \norm{f-h_{U_D^{mf}}}_\infi\le\eps .
\end{equation}

Second, apply \(\Cref{cor:dense-circuit-upper}\) with accuracy \(\eps/2\). This gives
\begin{equation}
    \norm{h-h_{U_D^{\mathrm{dense}}}}_\infi
    \le
    \frac{\eps}{2}
\end{equation}
with the same asymptotic dense-upload count. The triangle inequality gives
\begin{equation}
    \norm{f-h_{U_D^{\mathrm{dense}}}}_\infi\le\eps .
\end{equation}
The density statement follows by first choosing a tunable upload approximant within \(\eps/2\), then applying the appropriate coordinate or dense fixed upload replacement. The final assertion is the special case \(\mathcal F=C_{\le1}(X)\).
\end{proof}

\phantomsection\label{proof:mismatch-function-lower}
\begin{proof}[Proof of \Cref{cor:mismatch-function-lower}]
Let
\begin{equation}
    r(t)=h_D(t)-h(t).
\end{equation}
On \([0,1]\), \(\abs{r(t)}\le\eps\). Since \(h_D\) is 4-periodic,
\begin{equation}
    h_D(4)=h_D(0).
\end{equation}
Therefore
\begin{equation}
\begin{aligned}
    \abs{r(4)}
    &=
    \abs{h_D(0)-h(4)}    \\
    &\ge
    \abs{h(0)-h(4)}
    -
    \abs{h_D(0)-h(0)}    \\
    &\ge
    \Delta_4(h)-\eps
    \ge
    \frac{\Delta_4(h)}{2}.
\end{aligned}
\end{equation}
A scalar amplitude hypothesis from a depth-\(D\) fixed upload circuit has at most \(2D+1\) 4-periodic exponential terms, and \(h\) has at most \(S\) terms, so \(r\) has at most \(2D+1+S\) terms. Apply \(\Cref{thm:turan-nazarov}\) with \(I=[0,4]\) and \(E=[0,1]\), absorbing the fixed ratio \(\abs I/\abs E=4\) into a universal constant \(\Lambda>1\). Then
\begin{equation}
    \frac{\Delta_4(h)}{2}
    \le
    \sup_{t\in[0,4]}\abs{r(t)}
    \le
    \Lambda^{2D+S}\eps .
\end{equation}
Taking logarithms and rearranging proves the claim.
\end{proof}

\phantomsection\label{proof:dense-family-lower}
\begin{proof}[Proof of \Cref{cor:dense-family-lower}]
Let
\begin{equation}
    r(t)
    =
    h_{U_D^{\mathrm{dense}}}(\gamma(t))-h^*(\gamma(t)).
\end{equation}
The assumed uniform approximation on \([0,1]^m\) gives
\begin{equation}
    \abs{r(t)}\le\eps,
    \qquad 0\le t\le1.
\end{equation}

Along the slice \(\gamma(t)=x^{(0)}+te_k\), each dense upload
\begin{equation}
    V_m(\gamma(t))
    =
    \sum_{j=1}^{m}
    \ket{j}\!\bra{j}\otimes e^{i\pi \gamma_j(t)\sigma_z/4}
\end{equation}
has \(t\)-dependent phases only in the \(k\)-th block, where the scalar phases are
\(
    e^{\pm i\pi t/4}.
\)
All other blocks contribute only \(t\)-independent factors. Therefore a product with at most \(D\) dense uploads has scalar amplitudes of the form
\begin{equation}
    h_{U_D^{\mathrm{dense}}}(\gamma(t))
    =
    \sum_{n=-D}^{D}b_ne^{i\pi nt/4}.
\end{equation}
In particular, \(h_{U_D^{\mathrm{dense}}}\circ\gamma\) is period four and has at most \(2D+1\) exponential terms.

Using period four of the dense fixed upload amplitude on this slice,
\begin{equation}
\begin{aligned}
    \abs{r(8)}
    &=
    \abs{
    h_{U_D^{\mathrm{dense}}}(\gamma(0))
    -
    h^*(\gamma(8))
    }    \\
    &\ge
    \abs{h^*(\gamma(0))-h^*(\gamma(8))}
    -
    \abs{
    h_{U_D^{\mathrm{dense}}}(\gamma(0))
    -
    h^*(\gamma(0))
    }    \\
    &\ge
    \Delta_{8,\gamma}(h^*)-\eps
    \ge
    \frac{\Delta_{8,\gamma}(h^*)}{2}.
\end{aligned}
\end{equation}

By assumption, \(h^*\circ\gamma\) has at most \(S\) exponential terms. Hence \(r\) has at most \(2D+1+S\) terms. Apply \(\Cref{thm:turan-nazarov}\) with \(I=[0,8]\) and \(E=[0,1]\), absorbing the fixed ratio \(\abs I/\abs E=8\) into a universal constant \(\Lambda>1\). Then
\begin{equation}
    \frac{\Delta_{8,\gamma}(h^*)}{2}
    \le
    \sup_{t\in[0,8]}\abs{r(t)}
    \le
    \Lambda^{2D+S}\eps .
\end{equation}
Taking logarithms and rearranging proves the stated bound. The final claim follows by applying this bound to the single member \(h^*\). 
\end{proof}
\bibliography{bibliography}

\end{document}